\documentclass[12pt]{article}

\usepackage{amsmath,amsfonts,amssymb,amsthm,graphicx,cite}
\usepackage{amscd}
\usepackage{mathrsfs}
\usepackage{undertilde}
\usepackage[all,cmtip]{xy}

\numberwithin{equation}{section}

\begin{document}

\newtheorem{theorem}{Theorem}[section]
\newtheorem{corollary}[theorem]{Corollary}
\newtheorem{lemma}[theorem]{Lemma}
\newtheorem{proposition}[theorem]{Proposition}

\newcommand{\adiffop}{A$\Delta$O}
\newcommand{\adiffops}{A$\Delta$Os}

\newcommand{\be}{\begin{equation}}
\newcommand{\ee}{\end{equation}}
\newcommand{\bea}{\begin{eqnarray}}
\newcommand{\eea}{\end{eqnarray}}
\newcommand{\sh}{{\rm sh}}
\newcommand{\ch}{{\rm ch}}
\newcommand{\einde}{$\ \ \ \Box$ \vspace{5mm}}
\newcommand{\De}{\Delta}
\newcommand{\de}{\delta}
\newcommand{\Z}{{\mathbb Z}}
\newcommand{\N}{{\mathbb N}}
\newcommand{\C}{{\mathbb C}}
\newcommand{\Cs}{{\mathbb C}^{*}}
\newcommand{\R}{{\mathbb R}}
\newcommand{\Q}{{\mathbb Q}}
\newcommand{\T}{{\mathbb T}}
\newcommand{\re}{{\rm Re}\, }
\newcommand{\im}{{\rm Im}\, }
\newcommand{\cW}{{\cal W}}
\newcommand{\cJ}{{\cal J}}
\newcommand{\cE}{{\cal E}}
\newcommand{\cA}{{\cal A}}
\newcommand{\cR}{{\cal R}}
\newcommand{\cP}{{\cal P}}
\newcommand{\cM}{{\cal M}}
\newcommand{\cN}{{\cal N}}
\newcommand{\cI}{{\cal I}}
\newcommand{\cMs}{{\cal M}^{*}}
\newcommand{\cB}{{\cal B}}
\newcommand{\cD}{{\cal D}}
\newcommand{\cC}{{\cal C}}
\newcommand{\cL}{{\cal L}}
\newcommand{\cF}{{\cal F}}
\newcommand{\cH}{{\cal H}}
\newcommand{\cS}{{\cal S}}
\newcommand{\cT}{{\cal T}}
\newcommand{\cU}{{\cal U}}
\newcommand{\cQ}{{\cal Q}}
\newcommand{\cV}{{\cal V}}
\newcommand{\cK}{{\cal K}}
\newcommand{\intR}{\int_{-\infty}^{\infty}}
\newcommand{\intI}{\int_{0}^{\pi/2r}}
\newcommand{\limp}{\lim_{\re x \to \infty}}
\newcommand{\limn}{\lim_{\re x \to -\infty}}
\newcommand{\limpn}{\lim_{|\re x| \to \infty}}
\newcommand{\diag}{{\rm diag}}
\newcommand{\Ln}{{\rm Ln}}
\newcommand{\Arg}{{\rm Arg}}

\title{Joint eigenfunctions for the relativistic Calogero-Moser Hamiltonians of hyperbolic type. I. First steps}
\author{Martin Halln\"as \\Department of Mathematical Sciences, \\ Loughborough University, Leicestershire LE11 3TU, UK \\ and \\Simon Ruijsenaars \\ School of Mathematics, \\ University of Leeds, Leeds LS2 9JT, UK}

\date{}

\maketitle

\begin{abstract}
We present and develop a recursion scheme to construct joint eigenfunctions for the commuting analytic difference operators associated with the integrable $N$-particle systems of hyperbolic relativistic Calogero-Moser type. The scheme is based on kernel identities we obtained in previous work. In this first paper of a series we present the formal features of the scheme, show explicitly its arbitrary-$N$ viability for the `free' cases, and supply the analytic tools to prove the joint eigenfunction properties in suitable holomorphy domains.  
\end{abstract}

\tableofcontents


\section{Introduction}

This paper is the first in a series of articles that are concerned with the explicit diagonalization and Hilbert space transform theory for the relativistic generalization of the hyperbolic $N$-particle Calogero-Moser system. As is well known, the nonrelativistic version is defined by the Hamiltonian
\be\label{H2}
H_2=-\frac{\hbar^2}{2}\sum_{j=1}^N\partial_{x_j}^2+g(g-\hbar)\sum_{1\le j<l\le N}\mu^2/4\sinh^2(\mu(x_j-x_l)/2),
\ee
where $\hbar$ is Planck's constant, $g>0$ a coupling constant with dimension [action], and~$\mu>0$ a parameter with dimension [position]$^{-1}$. There exist $N-1$ additional independent PDOs~$H_k$ of order~$k$, $k=1,3,\ldots,N$, such that the PDOs form a commutative family. The simplest Hamiltonian is the momentum operator
\be
H_1=-i\hbar \sum_{j=1}^N\partial_{x_j},
\ee
but the remaining Hamiltonians will not be specified here. (They can be found for example in the survey~\cite{OP83}.)

The arbitrary-$g$ joint eigenfunctions of these PDOs were introduced and studied by Heckman and Opdam~\cite{HO87}, and their Hilbert space transform properties were obtained by Opdam~\cite{Opd95}. (More precisely, these authors handle arbitrary root systems, whereas we restrict attention to~$A_{N-1}$.) For the case $N=2$ the joint eigenfunction amounts to a specialization of the hypergeometric function $_2F_1$, and the associated Hilbert space transform is a special Jacobi function transform, cf.~Koornwinder's survey in~\cite{Koo84}.

The relativistic generalization~\cite{RS86,Rui87} is given by the commuting  analytic difference operators (henceforth A$\De$Os)
\be\label{Sk}
		S_k(x) = \sum_{\substack{I\subset\lbrace 1,\ldots,N\rbrace\\ |I|=k}}\prod_{\substack{m\in I\\ n\notin I}}f_-(x_m-x_n)\prod_{l\in I}\exp(-i\hbar \beta\partial_{x_l})\prod_{\substack{m\in I\\ n\notin I}}f_+(x_m-x_n),\ \ \ k=1,\ldots,N,
\ee
where
\be
f_{\pm}(z)= \big(\sinh(\mu (z\pm i\beta g )/2)/\sinh(\mu z/2)\big)^{1/2}.
\ee
Here, $\beta>0$ can be viewed as $1/mc$, with $m=1$ the particle mass and $c$ the speed of light. In the nonrelativistic limit $c\to\infty$ these operators give rise to the above commuting PDOs. (See~\cite{Rui94} for a survey of the relativistic Calogero-Moser systems and their various limits.) 

Thus far, only for~$N=2$ the eigenfunctions and Hilbert space transform are well understood. Indeed, they can be obtained by specializing results by the second-named author pertaining to a `relativistic' hypergeometric function generalizing~$_2F_1$~\cite{Rui99,Rui03II,Rui03III}. This function is defined in terms of an integral whose integrand involves solely products of the hyperbolic gamma function from~\cite{Rui97}. (See also~\cite{vdB06,BRS07} for other perspectives on this function.)

In recent years, novel integral representations of the pertinent $A_1$-type (one-coupling) specializations of the latter $BC_1$ (four-coupling) function have been obtained~\cite{Rui11}. These representations amount to Fourier transforms of products of hyperbolic gamma functions. 
For our purposes, the latter Fourier transform representations are of crucial importance. Indeed, as we shall show, they can be viewed as the result of the step from~$N=1$ to $N=2$ in a recursive construction of the arbitrary-$N$ joint eigenfunctions of the A$\De$Os~$S_k$~\eqref{Sk}. The point is that the plane wave in the integrand can be viewed as the $N=1$  eigenfunction, whereas the product of hyperbolic gamma functions serves as a kernel function, connecting the free $N=1$ A$\De$O~$\exp(-i\beta\hbar d/dx)$ to the interacting $N=2$ A$\De$Os.

In a recent joint paper~\cite{HR11}, we presented a comprehensive study of kernel functions for all of the Calogero-Moser and Toda systems of $A_{N-1}$ type. In particular, we obtained kernel functions connecting the hyperbolic A$\De$Os for the $N$-particle case to those for the $(N-1)$-particle case (see also~\cite{KNS09}). For the case $N=2$, the pertinent kernel functions amount to those occurring in~\cite{Rui11}, and this enables us to set up a recursion scheme for arbitrary~$N$, as sketched in Section~2.

The idea that such a recursive construction might be feasible is not new. It appears to date back to work by Gutzwiller~\cite{Gut81}, who used it to connect eigenfunctions for the periodic and nonperiodic (nonrelativistic) Toda systems. This formalism was then used for several other cases, in particular by Kharchev, Lebedev and Semenov-Tian-Shansky~\cite{KLS02} for the relativistic Toda systems, and by Gerasimov, Kharchev and Lebedev~\cite{GKL04} for the $g=1/2$ specialization of the nonrelativistic hyperbolic Calogero-Moser system (cf.~\eqref{H2}) and for the nonrelativistic Toda systems. We would like to express our indebtedness to this previous work, without which the scheme might seem to come out of the blue. 

In the later work following Gutzwiller's pioneering contribution, the representation theory of Whittaker modules and Yangians plays a pivotal role, and accordingly a considerable algebraic machinery occurs. However, as will become clear for the present case, the formal structure of the recursion scheme can be understood quite easily. Indeed, in our approach the main algebraic input consists solely of the kernel identities from our previous paper, and the same reasoning applies to the nonrelativistic hyperbolic Calogero-Moser Hamiltonians and to the nonperiodic Toda  Hamiltonians, for which the pertinent kernel identities were also obtained in~\cite{HR11}. 

On the other hand, the simplicity of the construction in Section~2  hinges on glossing over a great many analytical difficulties. In fact, it is a major undertaking to show that the integrals yield meromorphic joint eigenfunctions that give rise to a unitary eigenfunction transform with the long list of expected symmetry properties. The snags at issue are already considerable for the first steps, and will become clear in due course. In this paper our focus is on a complete proof of the joint eigenfunction properties in suitable holomorphy domains, leaving various issues (including global meromorphy) open for now.
 
We plan to address the analogous problems for the nonrelativistic hyperbolic Calogero-Moser and nonperiodic Toda  systems in later papers. In particular, in the aforementioned work dealing with recursive eigenfunctions for these systems, the expected duality properties are not shown and unitarity (`orthogonality and completeness') is left open. Also, the associated scattering theory needs to be studied, so as to confirm the long-standing conjecture that the particles exhibit soliton scattering (conservation of momenta and factorization). These features are quantum analogs of classical ones exhibited by the action-angle transforms for these systems~\cite{Rui88}, and their relevance for the relation to the sine-Gordon quantum field theory has been discussed in~\cite{Rui01}.

There are two length scales in the relativistic Hamiltonians~\eqref{Sk}, which we reparametrize as
\be
a_{+}=2\pi/\mu,\ \ \ a_{-}=\hbar\beta.
\ee
Also, we trade the coupling $g$ for a new parameter $b$ with dimension [position], namely,
\be
b= \beta g.
\ee
With these replacements in~\eqref{Sk} in effect, the Hamiltonians with $a_{+}$ and $a_{-}$ interchanged commute with the given ones, since the shifts change the  arguments of the coefficients by a period. The resulting $2N$ commuting Hamiltonians can be rewritten as
\be\label{Hk}
	H_{k,\delta}(b;x) = \sum_{\substack{I\subset\lbrace 1,\ldots,N\rbrace\\ |I|=k}}\prod_{\substack{m\in I\\ n\notin I}}f_{\delta,-}(x_m-x_n)\prod_{l\in I}\exp(-ia_{-\delta}\partial_{x_l})\prod_{\substack{m\in I\\ n\notin I}}f_{\delta,+}(x_m-x_n),
\ee
where $k=1,\ldots,N$, $ \de=+,-$, and
\be
	f_{\delta,\pm}(z) = \left(\frac{s_\delta(z\pm ib)}{s_\delta(z)}\right)^{1/2}.
\ee
Here and throughout the paper, we use the functions
\be\label{sce}
s_{\de}(z)= \sinh(\pi z/a_{\de}),\ \ c_{\de}(z)= \cosh(\pi z/a_{\de}),\ \ e_{\de}(z)= \exp(\pi z/a_{\de}),\ \ \de=+,-.
\ee
It is also convenient to use the parameters
\be\label{aconv}
\alpha= 2\pi/a_+a_-,\ \ \ \ a= (a_++a_-)/2,
\ee 
\be\label{asl}
a_s=\min (a_+,a_-),\ \ \ a_l=\max (a_+,a_-),\ \ \ a_+,a_->0.
\ee

To be sure, there are a great many different Hamiltonians commuting with the given Hamiltonians $H_{1,+}\ldots ,H_{N,+}$. Indeed, one can replace the functions $f_{-,\pm}(z)$ in~$H_{k,-}(x)$ by {\it arbitrary} functions with period $ia_{-}$. Of course, in that case the resulting Hamiltonians~$H_1^{'},\ldots ,H_N^{'}$ will not pairwise commute any more. But the latter feature can be ensured by choosing any parameter $b'$ that differs from $b$, and then all of the~$2N$ Hamiltonians do commute. Unless $b'$ equals $2a-b$, however, it it extremely unlikely that joint eigenfunctions of the Hamiltonians $H_{k,+}$ and $H_k^{'}$ exist. The choice $b'=2a-b$ is exceptional, since it yields again the Hamiltonians $H_{k,-}$, as becomes clear by pushing the functions on the right of the shifts to the left, with arguments shifted accordingly.

From the perspective of Hilbert space (which we do not address in this paper), it is crucial to restrict the parameters to the set
\be\label{Pi}
\Pi \equiv \{ (a_+,a_-,b)\in (0,\infty)^3\mid b<2a \}.
\ee
Clearly, whenever the coupling $b$ is real, the $2N$ Hamiltonians~\eqref{Hk} are formally self-adjoint, viewed as operators on the Hilbert space~$L^2(\R^N,dx)$. To promote them to bona fide commuting self-adjoint operators, however, the restriction of $b$ to the bounded interval~$(0,2a)$ is already imperative for~$N=2$, since this key feature is violated for generic~$b>2a$, cf.~\cite{Rui00}.  

At this point we are in the position to add some further remarks about related literature. First, there is Chalykh's paper~\cite{Cha02}, where Baker-Akhiezer type eigenfunctions of the above $N$-particle Hamiltonians are introduced. With our conventions, these correspond to the special $b$-choices $b=ka_{+}$ or $b=ka_-$ with $k$ integer. We intend to clarify the relation of the arbitrary-$b$ eigenfunctions furnished by the present method to Chalykh's eigenfunctions in later work. 

Secondly, there are several papers where so-called Harish-Chandra series solutions of the joint eigenvalue problem for the Hamiltonians are studied. A comprehensive study along these lines with extensive references is the recent paper by van Meer and Stokman~\cite{VMS09}. Roughly speaking, in this setting one arrives at eigenfunctions that correspond to only one of the two modular parameters
\be 
q_+=\exp(i\pi a_+/a_-),\ \ q_-=\exp(i\pi a_-/a_+),
\ee
on which our eigenfunctions depend in a symmetric way, in the sense that there is dependence on a single parameter $q$ that must have modulus smaller than one. Accordingly, one only considers the above Hamiltonians~$H_{k,\de}$ for one choice of $\de$. (For the $BC_1$ case the relation between the latter type of eigenfunction and the modular-invariant relativistic hypergeometric function has been clarified in~\cite{BRS07}.)

We have occasion to use two further incarnations of the Hamiltonians~$H_{k,\de}$, viewed again as acting on analytic functions. These are obtained by similarity transformation with a weight function and a scattering function. The latter are defined in terms of the hyperbolic gamma function~$G(a_+,a_-;z)$, whose salient features are summarized in Appendix~A. 

First, we define the generalized Harish-Chandra function
\be\label{cz}
c(b;z)=G(z+ia-ib)/G(z+ia),
\ee
and its multivariate version
\be\label{Cx}
C(b;x)=\prod_{1\le j<k\le N}c(b;x_j-x_k).
\ee
(Here and below, we usually suppress the dependence on the parameters~$a_+$ and $a_-$; the dependence on~$b$ is often omitted as well.) Then the weight and scattering functions are defined by
\be\label{wW}
w(z)=1/c(z)c(-z),\ \ \ W(x)=1/C(x)C(-x),
\ee
\be\label{uU}
u(z)=-c(z)/c(-z),\ \ \ U(x)=(-)^{N(N-1)/2}C(x)/C(-x).
\ee
Now we introduce
\be\label{Ak}
A_{k,\delta}(x) \equiv W(x)^{-1/2}H_{k,\delta}(x)W(x)^{1/2},\ee
\be\label{cAk}
\cA_{k,\de}(x)\equiv U(x)^{ -1/2}H_{k,\de}(x)U(x)^{ 1/2}=C(x)^{-1}A_{k,\de}(x)C(x),
\ee
where $k=1,\ldots,N$, and $  \delta=+,-$.

From the difference equations~\eqref{Gades} satisfied by the hyperbolic gamma function it follows that we have
\be\label{Wade}
\frac{1}{W(x)}\prod_{m\in I}\exp(-ia_{-\delta}\partial_{x_m})W(x)=\prod_{\substack{m\in I\\ n\notin I}}\frac{f^{2}_{\delta,-}(x_m-x_n)}{f^{2}_{\delta,+}(x_m-x_n-ia_{-\de})}.
\ee
Hence we obtain
\be\label{Aks}
A_{ k,\delta}(x) =\sum_{\substack{I\subset\lbrace 1,\ldots,N\rbrace\\ |I|=k}}\prod_{\substack{m\in I\\ n\notin I}}\frac{s_{\delta}(x_m-x_n-ib)}{s_{\delta}(x_m-x_n)}\prod_{l\in I}\exp(- ia_{-\delta}\partial_{x_l}).
\ee
Likewise, we deduce that the second similarity transformation yields
\be\label{cAks}
\cA_{ k,\delta}(x) =\sum_{\substack{I\subset\lbrace 1,\ldots,N\rbrace\\  |I|=k}}\ \prod_{\substack{m\in I, n\notin I\\m>n}}\frac{s_{\de}(x_m-x_n-ib)}{s_{\de}(x_m-x_n)}\frac{s_{\de}(x_m-x_n+ib-ia_{-\de})}{s_{\de}(x_m-x_n-ia_{-\de})}\prod_{l\in I}\exp(- ia_{-\delta}\partial_{x_l}).
\ee
Thus the similarity-transformed A$\De$Os act on the space of meromorphic functions. For parameters in $\Pi$ and $x\in\R^N$, the weight function $W(x)$ is positive and the `$S$-matrix' $U(x)$ has modulus one. Accordingly, the A$\De$Os~$A_{k,\de}$ and $\cA_{k,\de}$ are then formally self-adjoint, viewed as operators on the Hilbert spaces~$L^2(\R^N,W(x)dx)$ and~$L^2(\R^N,dx)$, resp. (Once more, these features still hold for~$b$ real.)

The scattering function satisfies
\be
U(2a-b;x)=U(b;x),
\ee
whereas the weight function is not invariant under the reflection~$b\to 2a-b$. Therefore, the A$\De$Os $A_{k,\de}$ are not invariant, whereas we have
\be
\cA_{k,\de}(2a-b;x)=\cA_{k,\de}(b;x),\ \ \ k=1,\ldots,N,\ \ \de=+,-.
\ee
On the other hand, $W(x)$ is symmetric (invariant under arbitrary permutations), whereas $U(x)$ is not symmetric. (Indeed, $w(z)$ is even, whereas $u(-z)$ equals $1/u(z)$.) As a consequence, the A$\De$Os~$A_{k,\de}$ are symmetric, whereas the~$\cA_{k,\de}$ are not. (Note this can be read off from the restriction~$m>n$ in~\eqref{cAks}.) More precisely, these different behaviors hold true for $k\le N-1$, since we have
\be
H_{N,\de}=A_{N,\de}=\cA_{N,\de}=\prod_{l=1}^N\exp(- ia_{-\delta}\partial_{x_l}).
\ee

The $b$-choices $a_+$ and $a_-$ have a special status, inasmuch as they lead to constant coefficients in~$H_{k,\de}$ and~$\cA_{k,\de}$ (but not in~$A_{k,\de}$). Indeed,
we have
\be\label{free}
H_{k,\pm}(a_{\de};x)=\cA_{k,\pm}(a_{\de};x)=\sum_{\substack{I\subset\lbrace 1,\ldots,N\rbrace\\ |I|=k}}
\prod_{l\in I}\exp(- ia_{\mp}\partial_{x_l}),\ \ \ k=1,\ldots,N,\ \ \de=+,-,
\ee
in accordance with no scattering taking place:
\be
U(a_{\de};x)=1,\ \ \ \de=+,-.
\ee
(This follows from \eqref{cz} and \eqref{uU} by using~\eqref{Gades}.) For these free cases, we shall show that the recursion scheme gives rise to the multivariate sine transform, and all of the expected properties can be readily checked.
Even for these cases, however, some non-obvious identities emerge. This is because the kernel functions are already nontrivial for the free cases, and they are the building blocks of the eigenfunctions. 

We proceed to sketch the results and the organization of the paper in more detail. In Section~2 we introduce the relevant kernel functions and their salient features, and present the recurrence scheme in a formal fashion (i.e., not worrying about convergence of integrals, etc.). As will be seen, the key algebraic ingredient for getting the eigenvalue structure expected from the explicit solution of the classical theory~\cite{Rui88} is given by the following recurrence for the elementary symmetric functions~$S^{(M)}_k$ of~$M$ nonzero numbers $a_1,\ldots,a_M$:
\be\label{Srec}
S^{(M)}_k(a_1,\ldots,a_M)=a_M^k\Big(S^{(M-1)}_{k}(a_1/a_M,\ldots,a_{M-1}/a_M)+S^{(M-1)}_{k-1}(a_1/a_M,\ldots,a_{M-1}/a_M)\Big).
\ee
Here we have $M\ge 1$, $k=1,\ldots,M,$ and
\be
S^{(M-1)}_M\equiv 0,\ \ \ S^{(M-1)}_0\equiv 1.
\ee

Section~3 is concerned with the free cases $b=a_{\pm}$. For these cases the kernel functions reduce to hyperbolic functions, for which it is feasible to evaluate the relevant integrals explicitly. (A key auxiliary integral is relegated to Appendix~C, cf.~Lemma~C.1.) Thus, the joint eigenfunctions can be obtained recursively, yielding the expected results.

In Section~4 we focus on the analytic aspects of the first step of the scheme, allowing $b$ in the strip $\re b\in (0,2a)$. This step leads from the free one-particle plane wave eigenfunction to the interacting two-particle eigenfunction, and yields the relativistic conical function from~\cite{Rui11} (after removal of the center-of-mass factor). We reconsider some properties of this function, using arguments that do not involve the previous representations of the $BC_1$ case (for which no multivariate recurrence is known). 

More specifically, we focus on holomorphy domains, the joint eigenvalue equations, and uniform decay bounds, with our reasoning (as laid down in Props.~4.1--4.4) serving as a template for the $N>2$ case. In this special case, however, we can proceed much further than for $N>2$. More precisely, we can easily obtain a larger holomorphy domain and corresponding bounds (cf.~Props.~4.5 and~4.6), since contour deformations do not lead to significant complications.

Section~5 is devoted to the step from $N=2$ to $N=3$. This leads to novel difficulties, but the counterparts of Props.~4.1--4.4 can still be proved. To control contour deformations, however, is already a quite arduous task for $N=3$, and we cannot easily get the expected holomorphy features in this way. (In later papers we hope to clarify the global meromorphy character in both $x$ and $y$ for arbitrary $N$ in other ways.) We do extend the holomorphy domain in the variable $x$ (as detailed in Prop.~5.5), but the present method seems too hard to push through for arbitrary $N$.

On the other hand, once our arguments yielding Props.~5.1--5.4 are well understood, the remaining difficulties for the inductive step treated in Section~6 are largely of a combinatoric and algebraic nature. This relative simplicity hinges on the explicit evaluations of some key integrals,  cf.~Lemmas~C.2 and C.3. It came as an unexpected bonus of the free case study in Section~3 that the integrals arising there (as encoded in Lemma~C.1) suggested to aim for bounds involving the related integrals of Lemmas~C.2 and C.3. Indeed, the latter furnish the tools to control the inductive step, which is encapsulated in  Propositions 6.1--6.4.

As we have mentioned already, this paper is the first in a series of articles. In Section 7 we provide a brief outlook on future work. We discuss the main aspects of the joint eigenfunctions that we plan to investigate, and also mention some of the results we expect.

In Appendix~A we review features of the hyperbolic gamma function we have occasion to use. In Appendix~B we derive bounds on the $G$-ratio featuring in the kernel functions, and on the weight function building block~$w(z)$, cf.~\eqref{cz}--\eqref{wW}.

As already mentioned, the explicit integrals needed to handle the free cases in Section~3 can be exploited to  explicitly evaluate two further integrals that are of crucial importance for the method we use to render the scheme rigorous. Indeed, the latter integrals enable us to derive in a recursive fashion uniform decay bounds on the joint eigenfunctions, which are needed to control the inductive step. Lemmas~C.1--C.3 contain the statements and proofs of the pertinent integrals.


\section{Formal structure of the recursion scheme}

We begin this section by detailing the various kernel functions and identities. First, the special function
\be\label{cS}
	\cS(b;x,y) \equiv \prod_{j,k=1}^N\frac{G(x_j-y_k-ib/2)}{G(x_j-y_k+ib/2)},
\ee
satisfies the kernel identities
 \be\label{AS}
	\big(A_{k,\delta}(x) - A_{k,\delta}(-y)\big)\cS(x,y) = 0,\quad k=1,\ldots,N,\quad \delta=+,-,
\ee
so that the functions
\be\label{Psi}
\Psi(x,y)\equiv [W(x)W(y)]^{1/2}\cS(x,y),
\ee
and
\be\label{cK}
\cK(x,y)\equiv [C(x)C(-y)]^{-1}\cS(x,y),
\ee
satisfy
\be\label{HPsi}
\big(H_{k,\delta}(x) - H_{k,\delta}(-y)\big)\Psi(x,y) = 0,\quad  k=1,\ldots,N,\quad \delta=+,-,
\ee
\be\label{cAcK}
\big(\cA_{k,\delta}(x) - \cA_{k,\delta}(-y)\big)\cK(x,y) = 0,\quad  k=1,\ldots,N,\quad \delta=+,-.
\ee
(See~\cite{Rui06} for the proof of~\eqref{AS}; the elliptic regime handled there is easily specialized to the hyperbolic one.)

The kernel functions just defined connect the $N$-particle A$\De$Os to themselves. As we intend to show in a later paper, they have a rather surprising application to the study of the joint eigenfunctions produced by the scheme. The protagonists of the scheme, however, are kernel functions connecting the $N$-particle A$\De$Os to the $(N-1)$-particle A$\De$Os, obtained in~\cite{HR11}. They arise from the previous ones by first multiplying by a suitable plane wave and then letting $y_N$ go to infinity. From now on, the dependence of the A$\De$Os and kernel functions on $N$ shall be made explicit wherever ambiguities might otherwise arise, and we also use a superscript $\sharp$ to denote a kernel with first argument in~$\C^M$ and second one in~$\C^{M-1}$.

With these conventions in place, the kernel function
\be\label{cSf}
	\cS_N^{\sharp}(b;x,y) \equiv \prod_{j=1}^N\prod_{k=1}^{N-1}\frac{G(x_j-y_k-ib/2)}{G(x_j-y_k+ib/2)},\ \ \ N>1,
\ee
satisfies the key identities (cf.~Corollary~2.3 in~\cite{HR11})
\begin{multline}\label{key}
	 A^{(N)}_{k,\delta}(x_1,\ldots,x_N)\cS_N^{\sharp}(x,y)\\ = \big(A^{(N-1)}_{ k,\delta}(-y_1,\ldots,-y_{N-1})+A^{(N-1)}_{ k-1,\delta}(-y_1,\ldots,-y_{N-1})\big)\cS_N^{\sharp}(x,y),
\end{multline} 
where $k=1,\ldots,N$, $\de=+,-$, and
\be
A^{(N-1)}_{N,\de}(-y_1,\ldots,-y_{N-1})\equiv 0,\ \ \ A^{(N-1)}_{0,\de}(-y_1,\ldots,-y_{N-1})\equiv {\bf 1}.
\ee
Using notation that will be clear from context, we now set
\be\label{Psif}
\Psi_N^{\sharp}(x,y)\equiv [W_N(x)W_{N-1}(y)]^{1/2}\cS_N^{\sharp}(x,y),
\ee
\be\label{cKf}
\cK_N^{\sharp}(x,y)\equiv [C_N(x)C_{N-1}(-y)]^{-1}\cS_N^{\sharp}(x,y).
\ee
The counterparts of~\eqref{key} are then (cf.~also \eqref{Psi}--\eqref{cAcK})
\be\label{key2}
	 H^{(N)}_{k,\delta}(x)\Psi_N^{\sharp}(x,y) = \big(H^{(N-1)}_{ k,\delta}(-y)+H^{(N-1)}_{ k-1,\delta}(-y)\big)\Psi_N^{\sharp}(x,y),
\ee
\be\label{key3}
	 \cA^{(N)}_{k,\delta}(x)\cK_N^{\sharp}(x,y) = \big(\cA^{(N-1)}_{ k,\delta}(-y)+\cA^{(N-1)}_{ k-1,\delta}(-y)\big)\cK_N^{\sharp}(x,y),
\ee
where $k=1,\ldots,N$, $\de=+,-$, and
\be
H^{(N-1)}_{N,\de}(-y)=\cA^{(N-1)}_{N,\de}(-y)\equiv 0,\ \ \ H^{(N-1)}_{0,\de}(-y)=\cA^{(N-1)}_{0,\de}(-y)\equiv {\bf 1}.
\ee

We are now prepared to explain the `calculational' crux of the recursion scheme. Assume we have a joint eigenfunction~$J_{N-1}((x_1,\ldots,x_{N-1}),(y_1,\ldots,y_{N-1}))$ of the A$\De$Os $A^{(N-1)}_{k,\de}(x)$,  with eigenvalues given by
\be\label{recass}
A^{(N-1)}_{k,\de}(x)J_{N-1}(x,y)=S_k^{(N-1)}\big(e_{\de}(2y_1),\ldots,e_{\de}(2y_{N-1})\big)J_{N-1}(x,y),\ \ \ k=1,\ldots,N-1.
\ee
where $S_k^{(M)}(a_1,\ldots,a_M)$ denotes the elementary symmetric functions of the~$M$ numbers $a_1,\ldots,a_M$. Now consider the function~$J_N(x,y)$ with arguments $x,y\in\C^N$, formally given by
\be\label{JN}
J_N(x,y)=\frac{\exp\big(i\alpha y_N\sum_{j=1}^Nx_j\big)}{(N-1)!}\int_{\R^{N-1}}dzW_{N-1}(z) \cS_N^{\sharp}(x,z)J_{N-1}(z,(y_1-y_N,\ldots,y_{N-1}-y_N)).
\ee
At this stage we do not address the convergence of the integral, and we also assume that we can take the $x$-dependent shifts in the A$\De$Os under the integral sign. 

Acting with $A^{(N)}_{k,\delta}(x)$ on~$J_N$, we pick up a factor~$e_{\de}(2ky_N)$ upon shifting the A$\De$O through the plane wave up front (recall~\eqref{aconv}), after which we act on the kernel function and use~\eqref{key} with $y\to z$. Using formal self-adjointness on $L^2(\R^{N-1},W_{N-1}(z)dz)$ of the two A$\De$Os on the rhs, we now transfer their action to the factor $J_{N-1}$, noting that the argument $-z$ should then be replaced by $z$, since no complex conjugation occurs in~\eqref{JN}. As a consequence we can use our assumption~\eqref{recass}, the upshot being that we obtain 
\bea\label{receig}
A^{(N)}_{k,\delta}(x) J_N(x,y)  &  =  &  e_{\de}(2ky_N)\big[ S^{(N-1)}_k\big(e_{\de}(2(y_1-y_N)),\ldots,e_{\de}(2(y_{N-1}-y_N))\big)
\nonumber \\
&  +  &  S^{(N-1)}_{k-1}\big(e_{\de}(2(y_1-y_N)),\ldots,e_{\de}(2(y_{N-1}-y_N))\big)\big]J_N(x,y).
\eea
Invoking the symmetric function recurrence~\eqref{Srec},
the eigenvalue formula \eqref{receig} can now be rewritten as
\be\label{recstep}
A^{(N)}_{k,\delta}(x) J_N(x,y)=S^{(N)}_k(e_{\de}(2y_1),\ldots,e_{\de}(2y_N))J_N(x,y),\ \ \ \ k=1,\ldots,N.
\ee
Comparing this to our assumption~\eqref{recass}, we easily deduce that we have arrived at a recursive procedure to construct joint eigenfunctions. Indeed, we can start the recursion with the plane wave
\be\label{J1}
J_1(x,y)\equiv \exp (i\alpha xy),
\ee
which obviously satisfies
\be
A^{(1)}_{1,\pm}J_1(x,y)=e_{\pm}(2y)J_1(x,y),
\ee
and then proceed inductively to obtain joint eigenfunctions for arbitrary $N$.

\section{The free cases $b=a_{\pm}$}

The algebraic aspects of the procedure detailed in the previous section are easy to grasp and unassailable from a formal viewpoint. From the perspective of rigorous analysis, however, the scheme thus far
 has the advantage of theft over honest toil.
Indeed, the first step in the recursion already leads to some delicate issues, as we shall see in the next section. 

On the other hand, the $N=1$ starting point causes no difficulty. Specifically, upon multiplication by $(a_+a_-)^{-1/2}$, the function~\eqref{J1} yields the kernel of a unitary integral operator on $L^2(\R)$. Accordingly, the two analytic difference operators~$\exp(-ia_{\mp}d/dx)$ at issue can be promoted to commuting self-adjoint operators on $L^2(\R,dx)$, defined as the pullbacks of the multiplication operators $e_{\pm}(2y)$ under Fourier transformation.

Specializing to the case $b=a_+$, we shall show in this section that the recurrence can be performed explicitly, yielding a multivariate version of the Fourier transform. (As will be seen shortly, the case $b=a_-$ yields basically the same result.)

We begin by writing down the specializations of the relevant functions. Setting $b=a_+$ in \eqref{cSf}, we deduce from the A$\De$Es \eqref{Gades} satisfied by the hyperbolic gamma function that the kernel function $\cS_N^{\sharp}$ reduces to
\be\label{cSs}
\cS_N^{\sharp}(x,y)=\prod_{j=1}^N\prod_{k=1}^{N-1}\frac{1}{2c_-(x_j-y_k)}.
\ee
Similarly, we find that the weight function $W_N$ reduces to
\be\label{wWs}
W_N(x)=\prod_{1\le j< k\le N}4s_-(x_j-x_k)^2,
\ee
cf.~\eqref{cz}--\eqref{wW}.

Turning to the explicit implementation of the recurrence, we start with the $N=2$ case, assuming $x,y\in\R^2$ until further notice. Substituting \eqref{cSs}--\eqref{wWs} into \eqref{JN} for $N=2$, we find
\be\label{J2s}
J_2(x,y)=e^{i\alpha y_2(x_1+x_2)}\int_{\mathbb{R}}dz\exp\big(i\alpha z(y_1-y_2)\big)\prod_{j=1}^2\frac{1}{2c_-(x_j-z)}.
\ee
Changing variable $z\to a_-z/\pi$ and multiplying by $4\pi/a_-$, we see that for $x_1\ne x_2$ and $y_1\ne y_2$ the resulting integral is equal to the right-hand side of \eqref{integral}, with $N=1$ and with $t=\pi x/a_-$, $p_1=2(y_1-y_2)/a_+$. Invoking Lemma~C.1 and making use of the identities $\sigma(x)\cdot y=x\cdot\sigma^{-1}(y)$ and $(-)^\sigma=(-)^{\sigma^{-1}}$, we deduce that $J_2$ is given by
\be
J_2(x,y)=\frac{-ia_-}{4s_-(x_1-x_2)s_+(y_1-y_2)}\sum_{\sigma\in S_2}(-)^\sigma\exp\big(i\alpha x\cdot\sigma(y)\big),
\ee
where $x,y\in\R^2$ and $ x_1\ne x_2,  y_1\ne y_2$.

We claim that this structure persists in the general case, in the sense that
\be\label{assumption}
J_N(x,y)=\prod_{1\leq j<k\leq N}
\frac{-ia_-}{4s_-(x_j-x_k)s_+(y_j-y_k)}
\sum_{\sigma\in S_N}(-)^\sigma\exp\big(i\alpha x\cdot\sigma(y)\big),
\ee
where $x,y\in\R^N$, and $x_j\ne x_k$, $y_j\ne y_k$ for $j\ne k$.
To verify this claim, we proceed by induction on $N$. Since the case $N=2$ has just been shown, we assume \eqref{assumption} for $N\geq 2$ and prove its validity for $N\to N+1$.

From \eqref{JN} and \eqref{cSs}--\eqref{wWs}, we infer
\begin{multline}\label{JNp1}
J_{N+1}(x,y) =\frac{\exp(i\alpha y_{N+1}(x_1+\cdots+x_{N+1}))}{2^{N(N+1)}N!}\prod_{1\leq j<k\leq N}\frac{-ia_-}{s_+(y_j-y_k)}\\ \times\sum_{\sigma\in S_N}(-)^\sigma \int_{\mathbb{R}^N}dz\exp\left(i\alpha\sum_{j=1}^Nz_j(y_{\sigma(j)}-y_{N+1})\right)\frac{\prod_{1\leq j<k\leq N}s_-(z_j-z_k)}{\prod_{j=1}^{N+1}\prod_{k=1}^Nc_-(x_j-z_k)}.
\end{multline}
Now, we fix $\sigma\in S_N$ and consider the corresponding integral in \eqref{JNp1}. Changing variables $z_k\to a_-z_k/\pi$, $k=1,\ldots,N$, and multiplying by $(\pi/a_-)^N$, we arrive at the right-hand side of \eqref{integral} for $t=\pi x/a_-$ and $p_j=2(y_{\sigma(j)}-y_{N+1})/a_+$, $j=1,\ldots,N$. Hence, Lemma~C.1 yields
\bea
J_{N+1}(x,y) & = &\frac{1}{N!}\prod_{1\leq j<k\leq N+1}
\frac{-ia_-}{4s_-(x_j-x_k)s_+(y_j-y_k)}
\nonumber \\ & & \times \sum_{\sigma\in S_N}(-)^\sigma\sum_{\tau\in S_{N+1}}(-)^\tau \exp\big(i\alpha\tau(x)\cdot\sigma(y)\big).
\eea
Here, we have identified $\sigma\in S_N$ with the element in $S_{N+1}$ that acts by $\sigma$ on the first $N$ coordinates and leaves the last one invariant. Upon replacing $\tau(x)\cdot\sigma(y)$ by $x\cdot(\tau^{-1}\sigma)(y)$, substituting $\tau\to\sigma\tau^{-1}$, and then using $(-)^{\sigma\tau^{-1}}=(-)^\sigma(-)^\tau$, we reduce the double summation to
\be
N!\sum_{\tau\in S_{N+1}}(-)^\tau\exp\big(i\alpha x\cdot\tau(y)\big).
\ee
If we now take $\tau\to\sigma$, then we arrive at \eqref{assumption} for $N\to N+1$.

We now summarize and slightly extend our finding.

\begin{theorem}
Assume $b$ is equal to $a_+$ or to $a_-$. Then we have
\be\label{JNfree}
J_N(a_{\de};x,y)=\prod_{1\leq j<k\leq N}
\frac{-ia_{-\de}}{4s_{-\de}(x_j-x_k)s_{\de}(y_j-y_k)}
\sum_{\sigma\in S_N}(-)^\sigma\exp\big(i\alpha x\cdot\sigma(y)\big),\ \ \ \de=+,-,
\ee
where $x,y\in\R^N$, and $x_j\ne x_k, y_j\ne y_k$ for $j\ne k$.
\end{theorem}
\begin{proof}
We have just shown this equality for $\de=+$.  We need only interchange $a_+$ and $a_-$ throughout to obtain~\eqref{JNfree} for $\de=-$.
\end{proof}

Using the $G$-A$\De$Es~\eqref{Gades}, one easily checks that the definitions~\eqref{cz}--\eqref{Cx} entail
\be\label{Cfree}
C(a_{\de};x)=1/\prod_{1\leq j<k\leq N}
2is_{-\de}(x_j-x_k).
\ee
Thus it follows that the similarity transformations~\eqref{Ak}--\eqref{cAk} give rise to joint eigenfunctions of the Hamiltonians $H_{k,\de}(a_{\de};x)$ and $\cA_{k,\de}(a_{\de};x)$ that are proportional to the kernel of the multivariate sine transform,
\be
\Sigma (x,y)\equiv \sum_{\sigma\in S_N}(-)^\sigma\exp\big(i\alpha x\cdot\sigma(y)\big),
\ee
 in accord with their free character for $b=a_{\pm}$, cf.~\eqref{free}.


\section{The step from $N=1$ to $N=2$}

In this section we consider the $N=2$ case for $b$-values in the strip
\be\label{Sa}
S_a\equiv \{ b\in\C \mid \re b \in (0,2a)\}.
\ee
Fixing $y\in\R^2$ until further notice, we begin by studying  the integrand
\be\label{I2}
I_2(b;x,y,z)\equiv \cS^{\sharp}_2(b;x,z)J_1(z,y_1-y_2)=  \exp(i\alpha z(y_1-y_2))\prod_{j=1}^2\frac{G(x_j-z-ib/2)}{G(x_j-z+ib/2)},
\ee
arising in the first step $J_1\to J_2$, cf.~\eqref{JN}. In the $z$-plane it has upward/downward double sequences of $G$-poles at (cf.~\eqref{Gpo}--\eqref{Gze})
\be\label{Ipo}
z=x_j-ib/2+ia+z_{kl},\ \ z=x_j+ib/2-ia-z_{kl},\ \ \ j=1,2,\ \ \ k,l\in\N.
\ee 
Letting first $x\in\R^2$, it is clear that the integration contour $\R$ stays below/above these upward/downward  sequences.
From Lemma~B.1 we see that $I_2$ decays exponentially for $z\to\pm \infty$, so the function
\be\label{J2}
J_2(b;x,y)=\exp(i\alpha y_2(x_1+x_2))\int_{\R}dzI_2(b;x,y,z),\ \  \ b\in S_a,
\ \ \ x,y\in\R^2,  
\ee
is well defined.

We proceed to obtain holomorphy properties of $J_2$ in the variables $b$ and $x$, using arguments that allow a recursive generalization to arbitrary $N$. (For the $N=2$ case, it is easy to obtain more information, as will transpire after Prop.~4.4.) First, from the fixed contour representation~\eqref{J2} and pole locations~\eqref{Ipo}, it is immediate that
$J_2$ is holomorphic in~$b$ and $x_j$ for $b\in S_a$ and $|\im x_j|<a-\re b/2$, $j=1,2$. Secondly, Lemma~B.1 entails 
\be\label{S2b}
\cS^{\sharp}_2(b;x,z)=O(\exp(-\alpha \re b\,|\re z|)),\ \ \ |\re z|\to\infty,
\ee
where the implied constant is uniform for $( b, x,\im z)$ vaying over compact subsets of $S_a\times\C^2\times\R$. 
 Thus we can freely shift the contour $\R$ up and down as long as we do not meet any poles. 

In particular, for $x\in\R^2$ we may replace the contour $\R$ in~\eqref{J2} by a contour $\R-ia+ib/2+i\epsilon$, $\epsilon\in(0,a-\re b/2)$. Then we can read off that $J_2$ extends holomorphically to $\im x_j\in (\epsilon,- 2a+\re b+\epsilon)$, $j=1,2$. Clearly, this argument can be iterated. Likewise, we can move the contour up step by step to allow arbitrary positive $\im x_j$ with similar restrictions. More precisely, introducing the domain
\be\label{D2}
D_2\equiv \{ (b,x)\in S_a\times \C^2 \mid |\im (x_1-x_2)|<2a-\re b\},
\ee
we have the following result.

\begin{proposition}
Let $y\in\R^2$. Then the function $J_2(b;x,y)$ is holomorphic in~$D_2$. 
\end{proposition}
\begin{proof}
We fix $(b,x)\in D_2$ and set
\be
\eta\equiv \im (x_1+x_2)/2.
\ee
Then we have
\be
\im x_j+a-\re b/2>\eta,\ \ \im x_j-a+\re b/2<\eta,\ \ \ j=1,2,
\ee
so the contour $\R +i\eta$ stays below/above the upward/downward  $G$-pole sequences~\eqref{Ipo}. Therefore, we can arrive at these $x$-values by the above iterative contour shift procedure without passing any poles. Hence we deduce holomorphy near $(b,x)$, and thus in~$D_2$.
\end{proof}

The analyticity properties obtained thus far already suffice to show that~$J_2$ is a joint eigenfunction of the $N=2$ A$\De$Os with the expected eigenvalues. More precisely, we fix attention on the subset 
\be\label{D2s}
D_{2,s}\equiv \{ (b,x)\in D_2 \mid  (b,(x_1-i\eta_1,x_2-i\eta_2))\in D_2,\forall (\eta_1,\eta_2)\in [0,a_l]^2\},\ \ \ \re b\in (0,a_s),
\ee
for which the argument shifts stay in the holomorphy domain $D_2$. Here we have used the notation~\eqref{asl}, noting that the restriction of $\re b$ to $(0,a_s)$ is necessary for $D_{2,s}$ to be non-empty. (Once we obtain further analyticity properties in $b$ and $x$, the eigenvalue equations can be analytically continued to $b\in S_a$, cf.~below.)

\begin{proposition}
Let $y\in\R^2$. For all $(b,x)\in D_{2,s}$, we have the joint eigenfunction property
\be\label{A2eig}
A^{(2)}_{k,\delta}(x) J_2(x,y)=S^{(2)}_k(e_{\de}(2y_1),e_{\de}(2y_2))J_2(x,y),\ \ k=1,2,\ \ \de=+,-.
\ee
\end{proposition}
\begin{proof}
Thanks to analyticity of $J_2$ in $D_2$ and the restriction to~$D_{2,s}$, the shifts of $x_j$ by $- i a_{-\de}$ are well defined. Also, by analyticity in the strip~$\re b\in (0,a_s)$, it suffices to prove~\eqref{A2eig} for $b\in(0,a_s)$. Then the distance~$a_s+a_l-b$ between the upward and downward pole sequences in~\eqref{Ipo} is larger than~$a_l$ by $\epsilon \equiv a_s-b>0$. Moreover, by analyticity in~$x$, we need only show~\eqref{A2eig} for~$|\im x_j|<\epsilon/2$, $j=1,2$ (say).

Letting $x$ vary over the region just detailed, we now proceed in a number of steps to rewrite the lhs of~\eqref{A2eig}. First, we note that for this $x$-region the $z$-contour~$\R$ in~\eqref{J2} is not only permissible, but can be moved down to
\be\label{C-}
C_-\equiv \R-ia_l/2,
\ee
without meeting any poles. Second, after shifting the A$\De$O through the plane wave up front (picking up a factor~$e_{\de}(2ky_2)$), we are allowed to act with it under the integral sign, since the poles stay away from~$C_-$ while they are shifted down by a distance~$a_{-\de}$. Thus the lhs equals
\be\label{lhs1}
e_{\de}(2ky_2)e^{i\alpha y_2(x_1+x_2)}\int_{C_-}dzA^{(2)}_{k,\delta}(x)S^{\sharp}_2(x,z)J_1(z,y_1-y_2).
\ee
 Third, we use the kernel identity~\eqref{key} to trade the action of $ A^{(2)}_{k,\de}(x)$ on~$S_2^{\sharp}(x,z)$ for that of~$\exp(ia_{-\de}d/dz)$ for $k=2$ and of~$\exp(ia_{-\de}d/dz) +{\bf 1}$ for~$k=1$. Then the integral in~\eqref{lhs1} can be written as
\be\label{lhs2}
\int_{C_-}dzS^{\sharp}_2(x,z+ia_{-\de})J_1(z,y_1-y_2)+(2-k)\int_{C_-}dzI_2(x,y,z).
\ee
Fourth, we change variables~$z\to z-ia_{-\de}$ in the first integral of~\eqref{lhs2}, which yields
\be\label{lhs3}
\int_{C_-+ia_{-\de}}dzS^{\sharp}_2(x,z)J_1(z-ia_{-\de},y_1-y_2)+(2-k)\int_{C_-}dzI_2(x,y,z).
\ee
 Fifth, we use the eigenvalue equation (recall~\eqref{J1})
\be
J_1(z-ia_{-\de},y_1-y_2)=e_{\de}(2(y_1-y_2))J_1(z,y_1-y_2),
\ee
and pull out the eigenvalue from the integration. Sixth and last, we shift the contours $C_-+ia_{-\de}$ and $C_-$ (for~$k=1$) back to~$\R$ without meeting poles, obtaining~\eqref{A2eig}. Thus the proof is complete.
\end{proof}

Next, let us take $z\to z+(x_1+x_2)/2$ in the integral~\eqref{J2}. This yields
\be\label{J2J}
J_2(b;(x_1,x_2),(y_1,y_2))=\exp(i\alpha (x_1+x_2)(y_1+y_2)/2)J(b;x_1-x_2,y_1-y_2),
\ee
where (recall the reflection equation~\eqref{refl})
\be\label{J}
J(b;u,v)\equiv \int_{\R}dze^{i\alpha zv}\prod_{\de_1,\de_2=+,-}G(\de_1z+\de_2 u/2-ib/2),\ \ \ \ u,v\in\R,\ \ \ b\in S_a.
\ee
From this it is clear that~$J$ is even in~$u$ and~$v$. Furthermore, it follows that~$J_2$ is invariant under the interchanges~$x_1\leftrightarrow x_2$ and~$y_1\leftrightarrow y_2$. Note that the first symmetry property is already plain from the defining formula~\eqref{J2}, whereas the second one is not immediate from~\eqref{J2}.

Comparing \eqref{J} to Eq.~(3.51) in~\cite{Rui11}, we see that $J$ is related to the relativistic conical function studied in~\cite{Rui11} via
\be\label{JcR}
J(b;x,y)=\sqrt{a_+a_-}G(ia-2ib)\cR(b;x,y)\prod_{\de=+,-}G(\de y-ia+ib).
\ee
From this a rather complete picture of $J$ and $J_2$ can be deduced. In order to achieve analytic control on the second step in the scheme, however, we need a uniform bound on the 2-particle eigenfunction~$J_2$, which exhibits its exponential decay for $|\re(x_1-x_2)|\to\infty$. This bound is not immediate from the results obtained in~\cite{Rui11}. 

\begin{proposition}
Let $y\in\R^2$.  For any $(b,x)\in D_2$ we have
\bea\label{J2b1}
|J_2(b;x,y)| & < & C(b,|\im (x_1-x_2)|)\exp(-\alpha \im (x_1+x_2)(y_1+y_2)/2)
\nonumber \\
& & \times\frac{ \re(x_1-x_2)}{\sinh(\gamma\,\re(x_1-x_2))},\ \ \ \ \ \ \gamma =\alpha\re b\,/2,
\eea
where $C$ is continuous on $S_a\times[0,2a-\re b)$. 
\end{proposition}
\begin{proof}
In view of~\eqref{J2J}, this amounts to a bound on $J(b;u,v)$.
Clearly, the representation~\eqref{J} can be analytically continued to the strip~$|\im u|<2a-\re b$, and for $u$ in this strip we have
\be\label{Jb}
|J(b;u,v)|\le \int_{\R}dz |r(b;z+u/2)r(b;z-u/2)|.
\ee
From Lemma~B.1 we now deduce
\be\label{JbC}
|J(b;u,v)|\le \int_{\R}dz
\prod_{\de=+,-}\frac{C_r(b,z-\de u/2)}{\cosh(\gamma( z-\de\re u/2))},
\ee
where the $C_r$-product  is bounded above by a function $C(b,|\im u|)$ that is continuous on $S_a\times [0,2a-\re b)$.
The remaining integral 
\be\label{rem}
 \int_{\R} \frac{dz}{\cosh(\gamma\, (z-\re u/2))\cosh(\gamma \,(z+\re u/2))},
\ee
can be easily evaluated by a residue calculation. (Alternatively, it follows from the $N=1$ case of Lemma~C.1 by letting $p$ converge to zero.) Hence we arrive at~\eqref{J2b1}.
\end{proof}

Thus far we have assumed that $y$ is real. As a consequence, the bound~\eqref{Jb} is independent of $v=y_1-y_2$. But when we combine~\eqref{J2J}--\eqref{J} with~\eqref{JbC}, it becomes obvious that $J_2$ has an analytic continuation to all $y\in\C^2$ satisfying $|\im (y_1-y_2)|<\re b$. Moreover, we can readily estimate $J_2$ in the larger holomorphy domain
\be
\cD_2\equiv \{ (b,x,y)\in S_a\times\C^2\times \C^2 \mid (b,x)\in D_2, |\im (y_1-y_2)|<\re b \}.
\ee

\begin{proposition}
For any $(b,x,y)\in \cD_2$ with $\im(y_1-y_2)\ne 0$ we have
\bea\label{J2b2}
|J_2(b;x,y)| & <  &C(b,|\im (x_1-x_2)|)\exp(-\alpha \im [(x_1+x_2)(y_1+y_2)]/2)
\nonumber \\
&  & \times \frac{\sinh(\alpha \im(y_1-y_2)\re(x_1-x_2)/2)}{\sin (\pi\im(y_1-y_2)/\re b)\sinh(\gamma\,\re(x_1-x_2))},
\eea
where $\gamma =\alpha\re b\,/2$ and $C$ is continuous on $S_a\times[0,2a-\re b)$. 
\end{proposition}
\begin{proof}
Once more, this is really a bound on $J(b;u,v)$, cf.~\eqref{J2J}--\eqref{J}. 
In this case we deduce from Lemma~B.1 that we have
\be\label{Jbv}
|J(b;u,v)|\le \int_{\R}dz\exp(-\alpha z \im v)
\prod_{\de=+,-}\frac{C_r(b,z-\de u/2)}{\cosh(\gamma( z-\de\re u/2))},
\ee
so now we obtain
\be\label{JbvC}
|J(b;u,v)|\le C(b,|\im u|) \int_{\R}dz \frac{\exp(-\alpha z\im v)}{\cosh(\gamma\, (z-\re u/2))\cosh(\gamma \,(z+\re u/2))},\ \ |\im v|<\re b,
\ee
where $C$ is continuous on $S_a\times [0,2a-\re b)$. The integral can be evaluated by a residue calculation, or follows directly from Lemma~C.1 with $N=1$. From this result we obtain the bound~\eqref{J2b2}.
\end{proof}

Clearly, we can view Prop.~4.3 as a corollary of Prop.~4.4, obtained by letting $\im (y_1-y_2)$ converge to 0. We have distinguished the two settings for expository reasons. Indeed, we need only invoke Prop.~4.3 in the next section to obtain the $N=3$ counterparts of Propositions~4.1 and 4.2, dealing with the joint eigenfunction properties for real $y\in\R^3$.  Then we obtain the $N=3$ analogs of Propositions~4.3 and 4.4, which suffice for the inductive step taken in Section~6. 

For the cases $N=2$ and $N=3$, we can readily derive further significant information, which we were unable to generalize inductively thus far. To start with, in the present $N=2$ case we can readily obtain further analyticity properties of~$J_2$, and we continue by doing so.

To this end, we fix $y\in\C^2$ with $|\im (y_1-y_2)|<\re b$. A moment's thought reveals that suitable contour deformations imply that $J_2$ has a holomorphic continuation to all~$(b,x)\in S_a\times\C^2$ such that $x_1-x_2+ib$  and $x_2-x_1+ib$ stay away from the `cut' $i[2a,\infty)$. Indeed, 
we need only indent the contour so as to ensure that the upward and downward pole sequences~\eqref{Ipo} remain above and below the contour, resp. Furthermore, contour pinching for $\pm \re(x_1-x_2)=\im b$ can be avoided by requiring $|\im (x_1-x_2)|<2a-\re b$.
We now summarize this analysis.

\begin{proposition}
The function~$J_2(b;x,y)$ is holomorphic
in $(b,x,y)$ on the extended domain
\be\label{cD2ex}
\cD_2^e\equiv \{ ( b,x,y)\in S_a\times\C^2\times \C^2 \mid  \pm (x_1-x_2) +ib \notin i [2a,\infty), |\im (y_1-y_2)|<\re b\}.
\ee
\end{proposition}

As announced, the A$\De$Es~\eqref{A2eig} can now be analytically continued to the subset 
\be\label{cD2s}
\cD_{2,s}^e\equiv \{ (b,x,y)\in\cD_2^e \mid  (b,(x_1-i\eta_1,x_2-i\eta_2),y)\in\cD_2^e,\forall (\eta_1,\eta_2)\in [0,a_l]^2\},
\ee
for which the argument shifts stay in the holomorphy domain $\cD_2^e$. In particular, we have
\be\label{D2eig}
 \{ (b,x,y)\in (0,2a)\times \C^2\times \R^2 \mid \re x_1\ne \re x_2\} \subset \cD_{2,s}^e.
\ee

Shifting the contour beyond the poles of the integrand, it is possible to show global meromorphy in~$x$, but to detail this is beyond the scope of the present paper. We do want to round off this section, however, by deriving a generalization of Prop.~4.4 to the extended holomorphy domain~$\cD_2^e$. By now it will be clear that this amounts to a further bound on~$J(b;u,v)$ on its holomorphy domain
\be
\cD_J\equiv \{ (b,u,v)\in S_a\times\C\times\C\mid \pm u +ib \notin i [2a,\infty), |\im v|<\re b\},
\ee
cf.~\eqref{J2J}--\eqref{J}. 

\begin{proposition}
For all $(b,u,v)\in\cD_J$ we have
\be\label{Jbex}
|J(b;u,v)|< C(b, u,v) \frac{\sinh(\alpha\im v\,\re u/2)}{\sin (\pi \im v/\re b)\sinh(\gamma\re u)}.
\ee
Here,  $C(b, u,v)$ is a continuous function on $\cD_J$, which satisfies
\be\label{C2as}
C(b,u,v)=O(1),\ \ \ |\re u|\to\infty,
\ee
where the implied constant can be chosen uniformly for $(b,\im u,\re v)$ varying over compact subsets of~$S_a\times\R\times \R$.
\end{proposition}
\begin{proof}
From holomorphy of $J$ in~$\cD_J$ it is plain that a bound of the form~\eqref{Jbex} is valid for~$|\re u|\le R$, with $R>0$ arbitrary and $C(b,u,v)$ continuous on~$\cD_J$. Fixing $K>0$ and choosing
\be
R\equiv |\im b|+2a+K,
\ee
(say), it remains to obtain an estimate~\eqref{Jbex} for~$b\in S_a$ and~$u\in\C$ satisfying~$|\re u|> R$ and~$|\im u|\le K$,
with $C(b,u,v)$ replaced by a function that only depends continuously on $b,K$ and $\re v$.

By evenness, we need only consider the case
$\re u> R$. Then we can represent $J$ as
\be\label{Jc}
J(b;u,v)=\int_{\cC}dz\exp(i\alpha zv)r(b;z)r(b;z-u),
\ee
where $\cC$ is a piecewise linear contour $\{ s+if_K(s)\mid s\in\R\}$, defined by setting
\be
f_K(s)=0, \ \ s\in (-\infty, 0]\cup [\re u+R,\infty),
\ee
\be
f_K(s)=\im u,\ \ s\in [\re u-|\im b|,\re u+|\im b|],
\ee
and letting $f_K(s)$ be a linear function so that $\cC$ connects the points~$0$, $u-|\im b|$ and~$u+|\im b|$, $\re u+R$ in the $z$-plane. Since we require $|\im u|\le K$, we have $|f_K'(s)|\le 1$ on~$\cC$. On the slanting parts of the contour, the increase in length for a given increase in the curve parameter~$s$ is therefore bounded above by $\sqrt{2}$ (compared to the flat pieces).
As a consequence, we deduce a majorization
\be
|J(b;u,v)|    \le  \sqrt{2} \exp(\alpha K|\re v|)
 \int_{\R}ds\exp(-\alpha s\im v) |r(b;s+if_K(s))r(b;s+if_K(s)-u)|.
\ee
 
 Next,
  invoking the $r$-bound~\eqref{rbc}--\eqref{Cras}, we readily infer that the product
  \be
C_r(b,z)C_r(b,z-u),\ \ \ z\in \cC,
\ee
can be bounded above by a constant depending only on $b$ and $K$. Thus the integral reduces to the one in~\eqref{JbvC}, and so the assertions follow.
\end{proof}


\section{The step from $N=2$ to $N=3$}

In this section we focus on the analytic aspects arising in the step $N=2\to N=3$, following the flow chart of Section~4 as far as possible. We recall the pertinent integrand is given by (cf.~\eqref{JN}) 
\be\label{I3}
I_3(b;x,y,z)=W_2(b;z)\cS^{\sharp}_3(b;x,z)J_2(b;z,(y_1-y_3,y_2-y_3)).
\ee
From Lemma~B.2 we can now obtain bounds on the 2-particle weight function defined by~\eqref{cz}--\eqref{wW}.
Likewise, Lemma~B.1 supplies bounds on the $G$-ratios in~$\cS_3^{\sharp}$, and Prop.~4.3 on~$J_2$.

More specifically, letting $y\in\R^3$ and $\im (z_1-z_2)=0$ until further notice, we obtain from~\eqref{rbc}
\be\label{Sb}
|\cS^{\sharp}_3(b;x,z)|<\prod_{j=1}^3\prod_{k=1}^2\frac{C_r(b,z_k-x_j)}{\cosh(\gamma \re(z_k-x_j))},\ \ \ \ \gamma=\alpha\re b/2,
\ee
provided $(b,z_k-x_j)\in \cD_r$ for $j=1,2,3$, and $k=1,2$; from~\eqref{wbr} we have
\be\label{W2be}
|W_2(b;z)|<C_w(b)\sinh^2(\gamma \, (z_1-z_2)),\ \ \ \im z_1=\im z_2,
\ee
whereas \eqref{J2b1} yields
\bea\label{J2be}
|J_2(b;z,(y_1-y_3,y_2-y_3))| & < & C(b)\exp(-\alpha \im (z_1+z_2)(y_1+y_2-2y_3)/2)
\nonumber  \\
& &  \times \frac{z_1-z_2}{\sinh(\gamma(z_1-z_2))},\ \ \  \ \ \ \ \im z_1=\im z_2.
\eea
It is clear from these bounds that~$I_3$ has exponential decay for $|\re z_1|\to \infty$ and for $|\re z_2|\to \infty$, uniformly for~$\im z_j$ varying over bounded intervals.

Requiring at first $x\in\R^3$, the integration contour $\R$ in the $z_k$-plane stays away from the $G$-poles at (cf.~\eqref{Gpo}--\eqref{Gze})
\be\label{I3po}
z_k=x_j-ib/2+ia+z_{lm},\ \ z_k=x_j+ib/2-ia-z_{lm},\ \ k=1,2,\ \  j=1,2,3,\ \ \ l,m\in\N,
\ee 
so that
\be\label{J3}
J_3(b;x,y)\equiv \frac{e^{i\alpha y_3(x_1+x_2+x_3)}}{2}\int_{\R^2}dzI_3(b;x,y,z),\ \  (b,x,y)\in S_a\times\R^3\times\R^3,  
\ee
is well defined.

It follows from the above that $J_3$ (as initially defined by~\eqref{J3}) extends to a function that is holomorphic in $b$ and $x_j$ for $b\in S_a$ and $|\im x_j|<a-\re b/2$, $j=1,2,3$. Furthermore, thanks to the above bounds, we can shift the two contours~$\R$ simultaneously up and down as long as we do not meet the poles~\eqref{I3po}. Using the same arguments as for the case~$N=2$, we can now extend the holomorphy domain step by step. In order to detail the result of this iterative procedure, we introduce
\be\label{D3}
D_3\equiv \{ (b,x)\in S_a\times \C^3 \mid \max_{1\le j<k\le 3}|\im (x_j-x_k)|<2a-\re b\}.
\ee
Also, given $x\in\C^3$, we set
\be\label{phix}
\phi(x)\equiv \im (x_{j_1}+x_{j_3})/2,
\ee
where the indices $j_1,j_3$ are such that
\be\label{imord}
\im x_{j_3}\le \im x_{j_2} \le \im x_{j_1},\ \ \ \{ j_1,j_2,j_3\}=\{ 1,2,3\}.
\ee
Then we have the following counterpart of Prop.~4.1.

\begin{proposition}
Let $y\in\R^3$. Then the function $J_3(b;x,y)$ is holomorphic in~$D_3$. 
\end{proposition}
\begin{proof}
Fixing $(b,x)\in D_3$, we clearly have
\be
\im x_j+a-\re b/2>\phi(x),\ \ \im x_j-a+\re b/2<\phi(x),\ \ \ j=1,2,3.
\ee
Therefore, the contour $\R +i\phi(x)$ stays below/above the upward/downward  $G$-pole sequences. Hence we can continue $J_3$ to any such $x$-value by simultaneous contour shifts, without passing any of the poles~\eqref{I3po} in the iteration. More specifically, we arrive at the associated representation
\be\label{J3eta} 
J_3(b;x,y)= \frac{e^{i\alpha y_3(x_1+x_2+x_3)}}{2}\int_{(\R+i\phi(x))^2}dzI_3(b;x,y,z),\ \  (b,x)\in D_3,\ \ y\in\R^3,  
\ee
from which holomorphy in~$D_3$ can be read off.
\end{proof}

Just as for the first step, these analyticity features are sufficient to show that $J_3$ is a joint eigenfunction of the six A$\De$Os at hand. Accordingly, we introduce the subdomain of $D_3$ given by
\be\label{D3s}
D_{3,s}\equiv  \{ (b,x)\in D_3\mid  (b,x-i\eta)\in D_3,\forall \eta\in [0,a_l]^3,\ \re b\in (0,a_s)\}.
\ee

On the other hand, we now have more than one contour, and in the proof of the following counterpart of Prop.~4.2 we have occasion to shift one of them. Thus, to show that this causes no problem with the contour tails, we can no longer appeal to the bounds~\eqref{W2be}--\eqref{J2be}, which require $\im z_1=\im z_2$. 

However, this snag is readily obviated, as follows. First, we note that the bounds~\eqref{wbc}--\eqref{Cwas} imply that the rate of divergence of $W_2(b;z)$ as $|\re(z_1-z_2)|\to\infty$ does not depend on~$\im(z_1-z_2)$. Second, we may invoke Prop.~4.3 as long as~$|\im(z_1-z_2)|<2a-\re b$, and on closed subintervals the decay rate in the bound on the $J_2$-factor in the integrand~$I_3$~\eqref{I3} as $|\re(z_1-z_2)|\to\infty$ again does not depend on~$\im(z_1-z_2)$. In view of the bound~\eqref{Sb} on the kernel function, this suffices for the contour shifts to be legitimate.

\begin{proposition}
Fixing $y\in\R^3$ and letting $(b,x)\in D_{3,s}$, we have the eigenvalue equations
\be\label{A3eig}
A^{(3)}_{k,\delta}(x) J_3(x,y)=S^{(3)}_k(e_{\de}(2y_1),e_{\de}(2y_2),e_{\de}(2y_3))J_3(x,y), \ \ k=1,2,3,\ \ \de=+,-.
\ee
\end{proposition}
\begin{proof}
We follow the same line of reasoning as in the proof of Prop.~4.2. Thus we infer again that we need only prove~\eqref{A3eig} for $b\in(0,a_s)$ and $|\im x_j|<\epsilon/2$, $j=1,2,3$, with $\epsilon=a_s-b$.

Letting $x$ vary over this region, we now adapt the six steps in the proof of Prop.~4.2. First, the $z_k$-contours~$\R$ in~\eqref{J3} are permissible and can be moved down to~$C_-$~\eqref{C-}
without meeting poles. Second, shifting the A$\De$O through the plane wave up front (yielding a factor~$e_{\de}(2ky_3)$), we are allowed to act with it under the integral sign, so that the lhs of~\eqref{A3eig} becomes
\be\label{3lhs1}
e_{\de}(2ky_3)\frac{e^{i\alpha y_3(x_1+x_2+x_3)}}{2}\int_{C_-}dz_1\int_{C_-}dz_2W_2(z)A^{(3)}_{k,\delta}(x)\cS^{\sharp}_3(x,z)J_2(z,(y_1-y_3,y_2-y_3)).
\ee
Third, we use~\eqref{key} to replace $A^{(3)}_{k,\de}(x)$ by the sum of $ A^{(2)}_{k,\de}(-z)$ and $ A^{(2)}_{k-1,\de}(-z)$. Then it remains to show that the integrals
\be\label{3lhs3}
\int_{C_-}dz_1\int_{C_-}dz_2W_2(z)J_2(z,(y_1-y_3,y_2-y_3))A^{(2)}_{l,\de}(-z)\cS^{\sharp}_3(x,z),
\ee
are equal to
\be\label{eigl}
S^{(2)}_l\big(e_{\de}(2(y_1-y_3)),e_{\de}(2(y_2-y_3))\big)\ 2\exp(-i\alpha y_3(x_1+x_2+x_3))J_3(x,y),
\ee
for $l=1$ and $l=2$.

Fourth, we change variables depending on the two cases at hand. For $l=2$ the A$\De$O-action yields~$\cS^{\sharp}_3(x,(z_1+ia_{-\de},z_2+ia_{-\de}))$ and we change variables~$z_k\to z_k-ia_{-\de}$, $k=1,2$, yielding the integral
\be
\int_{C_-+ia_{-\de}}dz_1\int_{C_-+ia_{-\de}}dz_2W_2(z)J_2((z_1-ia_{-\de},z_2-ia_{-\de}),(y_1-y_3,y_2-y_3))\cS^{\sharp}_3(x,z).
\ee
By virtue of Prop.~4.2 (with $k=2$), this is equal to
\be
e_{\de}(2(y_1-y_3)+2(y_2-y_3))\int_{C_-+ia_{-\de}}dz_1\int_{C_-+ia_{-\de}}dz_2W_2(z)J_2((z_1,z_2),(y_1-y_3,y_2-y_3))\cS^{\sharp}_3(x,z),
\ee
which amounts to the fifth step for this case.
Shifting the contours back to $\R$ (the sixth step), we see this equals~\eqref{eigl} for the case $l=2$.

It remains to show equality of~\eqref{3lhs3} and~\eqref{eigl} for $l=1$. In this case the integral~\eqref{3lhs3} is given by 
\bea\label{3lhs31}
&  & \int_{C_-}dz_1\int_{C_-}dz_2W_2(z)J_2(z,(y_1-y_3,y_2-y_3))
\nonumber \\
&  & \times\Big(\frac{s_{\de}(z_1-z_2+ib)}{s_{\de}(z_1-z_2)}\exp(ia_{-\de}\partial_{z_1})+\frac{s_{\de}(z_2-z_1+ib)}{s_{\de}(z_2-z_1)}\exp(ia_{-\de}\partial_{z_2})\Big)\cS^{\sharp}_3(x,z),
\eea
cf.~\eqref{Aks}. Note that the pole at $z_1=z_2$ due to the denominator $s_{\de}$ is matched by a double zero due to~$W_2(z)=w(z_1-z_2)$, cf.~\eqref{wW}.

Changing variables, we now rewrite~\eqref{3lhs31} as
\begin{multline}\label{rewr1}
\int_{C_-+ia_{-\de}}dz_1\int_{C_-}dz_2
\frac{s_{\de}(z_1-ia_{-\de}-z_2+ib)}{s_{\de}(z_1-ia_{-\de}-z_2)}W_2(z_1-ia_{-\de},z_2)
\\
\times \cS_3^{\sharp}(x,z)J_2((z_1-ia_{-\de},z_2),(y_1-y_3,y_2-y_3))
\\
+\int_{C_-}dz_1\int_{C_-+ia_{-\de}}dz_2
\frac{s_{\de}(z_2-ia_{-\de}-z_1+ib)}{s_{\de}(z_2-ia_{-\de}-z_1)}W_2(z_1,z_2-ia_{-\de})
\\
\times \cS_3^{\sharp}(x,z)J_2((z_1,z_2-ia_{-\de}),(y_1-y_3,y_2-y_3)).
\end{multline}
From the A$\De$Es~\eqref{Wade} satisfied by the $W$-function we deduce that this equals
\begin{multline}\label{rewr2}
\int_{C_-+ia_{-\de}}dz_1\int_{C_-}dz_2
\frac{s_{\de}(z_1-z_2-ib)}{s_{\de}(z_1-z_2)}W_2(z)
\\
\times \cS_3^{\sharp}(x,z)J_2((z_1-ia_{-\de},z_2),(y_1-y_3,y_2-y_3))
\\
+\int_{C_-}dz_1\int_{C_-+ia_{-\de}}dz_2
\frac{s_{\de}(z_2-z_1-ib)}{s_{\de}(z_2-z_1)}W_2(z)
\\
 \times \cS_3^{\sharp}(x,z)J_2((z_1,z_2-ia_{-\de}),(y_1-y_3,y_2-y_3)).
\end{multline}

We claim that when we shift the contour $C_-$ in these two integrals up by $a_{-\de}$, then no poles are met. (In view of the paragraph preceding~\eqref{A3eig}, we retain exponential decay on the contour tails, so the shift causes no problems at the tail ends.) Taking this claim for granted, the integrands are regular for $z_1,z_2\in C_-+ia_{-\de}$, and since the double contour in the first integral is now equal to the one in the second integral, we are entitled to use the eigenvalue equation~\eqref{A2eig} with $k=1$ (the fifth step). The sixth step is the same as for the case $l=2$ and now yields~\eqref{eigl} for $l=1$.

To complete the proof, it remains to verify the claim. Clearly, we stay in the holomorphy domains of $\cS_3^{\sharp}$ and $J_2$ while shifting the relevant contour. Setting~$z\equiv z_1-z_2$ and letting first $a_{-\de}=a_s$, we meet the simple poles of the factors $1/G(\pm z+ia-ib)$ at $\pm z=ib$, but they are matched by  zeros of $s_{\de}(\pm z-ib)$. Likewise, the simple poles of $s_{\de}(z)$ at $z=0$ and at $\pm z=ia_s$ (the latter being met when $a_s=a_l$) are matched by zeros of $G(\pm z+ia)$. Thus the first integrand is regular for $z_1\in C_-+ia_s$ and $\im (z_1-z_2)\in [0, a_s]$ and the second one for $z_2\in C_-+ia_s$ and $\im (z_2-z_1)\in [0, a_s]$.

Finally, consider the case $a_{-\de}=a_l$.  Then we meet simple poles of $1/G(\pm z+ia-ib)$ at $\pm z=ib+ika_s$ for $ka_s\le a_l-b$, but they are again matched by zeros of $s_{\de}(\pm z-ib)$. The poles of $s_{\de}(z)$ at $\pm z=ika_s$ for $ka_s\le a_l$ are matched by zeros of $G(\pm z+ia)$. Hence the first integrand is regular for $z_1\in C_-+ia_l$ and $\im (z_1-z_2)\in [0, a_l]$ and the second one for $z_2\in C_-+ia_l$ and $\im (z_2-z_1)\in [0, a_l]$. Thus our claim is proved.
\end{proof}

We proceed to obtain the analog of Prop.~4.3. Fixing $(b,x)\in D_3$, we recall the definition~\eqref{phix} of~$\phi(x)$, and define in addition (cf.~\eqref{imord})
\be\label{dx}
d(x)\equiv \im (x_{j_1}-x_{j_3}).
\ee

\begin{proposition}
Let $y\in\R^3$.  For any $(b,x)\in D_3$ we have
\bea\label{J3b1}
|J_3(b;x,y)| & < & C(b,d(x))\exp(-\alpha [\phi(x)(y_1+y_2)+y_3\im x_{j_2}])
\nonumber \\
&  &  \times \prod_{1\le m<n\le 3}\frac{ \re(x_m-x_n)}{\sinh(\gamma\,\re(x_m-x_n))},
\eea
where  $C$ is continuous on $S_a\times[0,2a-\re b)$. 
\end{proposition}

\begin{proof}
From \eqref{J3eta} we obtain a majorization
\be\label{J3maj}
|J_3(x,y)|\le \frac{\exp(-\alpha y_3\im(x_1+x_2+x_3))}{2}\int_{(\R+i\phi(x))^2}dz|W_2(z)\cS^{\sharp}_3(x,z)J_2(z,(y_1-y_3,y_2-y_3))|.
\ee
Combining Prop.~4.3 with the bounds~\eqref{Sb} and~\eqref{W2be}, we infer
\begin{multline}\label{J3bo}
|J_3(x,y)|  \le  C(b) \exp(-\alpha [y_3\im x_{j_2}+(y_1+y_2)\phi(x)])
 \\
\times \int_{\R^2}ds (s_1-s_2)\sinh(\gamma(s_1-s_2))
\prod_{j=1}^3\prod_{k=1}^2\frac{C_r(b,s_k+i\phi(x)- x_j)}{\cosh(\gamma (s_k-\re x_j))},
\end{multline}
with $C(b)$ continuous on $S_a$. Now we have (cf.~\eqref{phix}, \eqref{dx} and~\eqref{D3})
\be\label{phid}
 |\im(i\phi(x)-x_j)|\leq d(x)/2<a-\re b/2, \ \ \ j=1,2,3, \ \ x\in D_3.
 \ee
 Recalling~Lemma~B.1, it follows that
 the $C_r$-product  is bounded above by a function $C(b,d(x))$ that is continuous on $S_a\times [0,2a-\re b)$.
The remaining integral 
\be\label{rem2}
 \int_{\R^2}ds \frac{(s_1-s_2)\sinh(\gamma(s_1-s_2))}{\prod_{j=1}^3\prod_{k=1}^2\cosh(\gamma (s_k-\re x_j))},
\ee
 follows from the $N=2$ case of Lemma~C.2. Hence we arrive at~\eqref{J3b1}.
\end{proof}

Thus far we have assumed $y$ is real. When we fix $x\in\R^3$, however, it is easy to see from the defining formula~\eqref{J3} and Prop.~4.4, combined with the bounds~\eqref{Sb}--\eqref{W2be}, that $J_3$ is well defined for $y\in\C^3$ whose imaginary parts are sufficiently close. Indeed, we easily obtain an estimate
\begin{multline}
|J_3(x,y)|  <   \frac{C(b)}{\sin(\pi \im(y_1-y_2)/\re b)}\int_{\R^2}dz\exp(-\alpha (z_1+z_2)\im (y_1+y_2-2y_3)/2)
\\
  \times \frac{\sinh(\alpha(z_1-z_2))\sinh(\gamma \im (y_1-y_2)(z_1-z_2)/2)}{\prod_{j=1}^3\prod_{k=1}^2\cosh(\gamma (z_k- x_j))},
\end{multline}
and the integral occurring here converges when $|\im (y_1-y_2)|$ and~$|\im (y_1+y_2-2 y_3)|$ are smaller than $\re b$. This also entails that $J_3$ is holomorphic in~$y$ in the latter domain.

In fact, it turns out that a more refined bound gives rise to a larger $y$-domain, allowing also a subdomain of $D_3$ for the dependence on $(b,x)$. This hinges on Lemma~C.3, which enables us to calculate the integral occurring in a suitable majorization of $J_3(b;x,y)$. To arrive at the latter, we need to start from a quite different representation for $J_3$, which is only well defined on a subset of $D_3$. The details now follow. 

To begin with, we introduce the notation
\be\label{not3}
X_3\equiv \frac13\sum_{j=1}^3x_j,\ \ Y_3\equiv \frac13\sum_{j=1}^3y_j,\ \ \tilde{x}_j\equiv  x_j-X_3,\ \ \ 
\tilde{y}_j\equiv  y_j-Y_3,\ \ \ j=1,2,3,
\ee
and recall that the representation~\eqref{J3eta} is well defined for any~$(b,x)\in D_3$. By contrast, we need to restrict $(b,x)$ to
\be\label{D3r}
D_3^r\equiv \{ (b,x)\in S_a\times \C^3 \mid |\im \tilde{x}_j|<a-\re b/2,\ \ j=1,2,3 \}\subset D_3,
\ee
for the formula
\be\label{J3X} 
J_3(b;x,y)= \frac{e^{i\alpha y_3(x_1+x_2+x_3)}}{2}\int_{(\R+i\im X_3)^2}dzI_3(b;x,y,z),\ \ \ \ y\in\R^3,  
\ee
to be well defined and valid. This representation can be rewritten as
\be\label{J3s}
J_3(b;x,y)=\exp(3i\alpha X_3Y_3)J_3^r(b;x;y),
\ee
\be\label{J3r}
J_3^r(b;x,y)\equiv\frac{1}{2}\int_{\R^2}dsW_2(b;s)\cS^{\sharp}_3(b;\tilde{x},s)J_2(b;s,(y_1-y_3,y_2-y_3)),\ \ (b,x)\in D_3^r,\ \ y\in \R^3,
\ee
where we used~\eqref{J2J}. It is to be noted that $J_3^r$ depends on $x$ and $y$ only via the differences $x_1-x_2,x_2-x_3$ and $y_1-y_2,y_2-y_3$.

Next, we introduce the domain
\be\label{cD3}
\cD_3\equiv \{ (b,x,y)\in D_3^r\times \C^3 \mid  |\im (y_j-y_k)|<\re b,\ j,k=1,2,3 \},
\ee
and the maximum function
\be
\mu(x)\equiv \max_{j=1,2,3}|\im \tilde{x}_j|.
\ee
Now we are prepared for the analog of~Prop.~4.4.

\begin{proposition}
The function $J_3(b;x,y)$ is holomorphic in~$(b,x,y)$ on~$\cD_3$.
For any $(b,x,y)\in \cD_3$ with 
\be\label{xyres}
\re(x_j-x_k)\ne 0,\ \ \ \im(y_j-y_k)\ne 0, \ \ \ 1\le j<k\le 3,
\ee
 we have
\begin{multline}\label{J3by}
|J_3(b;x,y) | <  C(b,\mu(x))\exp(-3\alpha \im (X_3Y_3))
 \\
 \times \frac{\sum_{\tau\in S_3}(-)^{\tau}\exp\Big(-\alpha\sum_{j=1}^3\re \big(\tilde{x}_{\tau(j)}\big)\im \tilde{y}_j\Big)}{\prod_{1\le m<n\le 3}\sin (\pi\im(y_n-y_m)/\re b)\sinh(\gamma\,\re(x_m-x_n))},
\end{multline}
where  $C$ is continuous on $S_a\times[0,a-\re b/2)$. 
\end{proposition}
\begin{proof}
In view of~\eqref{J3s}, we need only estimate~$J_3^r$. 
From~Prop.~4.4 and the bounds~\eqref{Sb}--\eqref{W2be} we readily deduce
\begin{multline}\label{J3rb}
|J_3^r(x,y)|  <   \frac{C(b,\mu(x))}{\sin(\pi \im(y_1-y_2)/\re b)}\int_{\R^2}ds\exp(-\alpha (s_1+s_2)\im (y_1+y_2-2y_3)/2)
\\
  \times \frac{\sinh(\alpha(s_1-s_2))\sinh(\gamma \im (y_1-y_2)(s_1-s_2)/2)}{\prod_{j=1}^3\prod_{k=1}^2\cosh(\gamma (s_k- \re \tilde{x}_j))}. 
\end{multline}
We proceed to rewrite this in such a way that Lemma~C.3 may be invoked. To this end we set
 \be
S_2\equiv (s_1+s_2)/2, \ \ \tilde{s}_j\equiv s_j-S_2,\ \ j=1,2.
\ee
Consider now the auxiliary integral
\begin{multline}
A_2(t,u)\equiv \int_{\R^2}ds \frac{\sinh(\gamma(s_1-s_2))}{\prod_{j=1}^3\prod_{k=1}^2\cosh(\gamma s_k-t_j)}\exp\Big(\gamma S_2\sum_{j=1}^2 (u_3-u_j)\Big)
\\
\times\sum_{\tau \in S_2}(-)^{\tau}\exp\Big(\gamma\tilde{s}_{\tau(1)} (u_2-u_1)\Big).
\end{multline}
Clearly, $A_2(\gamma\re \tilde{x},2\im y/\re b)$ is equal to the integral in the bound~\eqref{J3rb} times a factor $-2$. The integrand is symmetric under the interchange $s_1\leftrightarrow s_2$, and the sum is antisymmetric. Thus we may replace the sum by twice the summand with~$\tau ={\rm id}$. Taking next $s\to z/\gamma$, this yields 
\begin{multline}
A_2(t,u)=\frac{2}{\gamma^2} 
\int_{\R^2}dz \frac{\sinh(z_1-z_2)}{\prod_{j=1}^3\prod_{k=1}^2\cosh (z_k- t_j)}
\\
\times \exp\Big(Z_2\sum_{j=1}^2(u_3-u_j)\Big)
\exp\Big((z_1-Z_2) (u_2-u_1)\Big).
\end{multline}
Comparing to~$C_2(u,t)$~\eqref{Cintegral}, we deduce
\be
A_2(t,u)=\frac{2}{\gamma^2}C_2(u,t).
\ee
Therefore, the bound~\eqref{J3rb} can be rewritten as
\be\label{J3C2}
|J_3^r(x,y)|  <   \frac{C(b,\mu(x))}{\gamma^2\sin(\pi \im(y_2-y_1)/\re b)}C_2(2\im y/\re b,\gamma\re \tilde{x}).
\ee

We can now appeal to Lemma~C.3 to deduce the holomorphy assertion. Also, assuming~\eqref{xyres}, we  can combine~\eqref{CN} with  
\bea\label{sums3}
\sum_{j=1}^2\re\big(\tilde{x}_{\tau(j)}\big)\im (y_3-y_j) & = &\sum_{j=1}^3\re\big(\tilde{x}_{\tau(j)}\big)\im (y_3-y_j)
\nonumber \\
=-\sum_{j=1}^3\re\big(\tilde{x}_{\tau(j)}\big)\im y_j & = &-\sum_{j=1}^3\re\big(\tilde{x}_{\tau(j)}\big)\im \tilde{y}_j,
\eea
to obtain~\eqref{J3by}.
\end{proof}

Thus far we have obtained holomorphy properties of $J_3$ by employing two contours equal to the same horizontal line. We conclude this section by allowing more general contours, so as to obtain a larger holomorphy domain. However, we still keep both contours equal and shall also restrict attention to real $y$. Indeed, as will transpire shortly, even with these restrictions there are novel complications arising for $N=3$ (as compared to $N=2$), and we shall not strive for optimal results. In fact, we do not try and obtain a counterpart of Prop.~4.6, since even with such a result available, it would be a major undertaking to extend the $N=4$ case beyond Props.~6.1--6.4. 

Just as for~$J_2$, we can deform the contours to show that $J_3$ remains holomorphic for $(b,x)$ in a larger domain than $D_3$. We restrict attention to deformations of the contours for which both contours remain equal to the same contour~$\cC$, defined by functions~$f(s)$ from the space $\cL_w$ delineated above~\eqref{freq}. The choice of contour is such that it  separates the upward and downward pole sequences~\eqref{I3po}.
Also, to stay clear of the $W_2$-poles for~$|\im b|>0$, we need~$f$ to be such that when~$|s|$ increases by~$|\im b|$, then~$|f(s)|$ increases by less than~$|\re b|$, cf.~\eqref{freq}. This can be achieved in two distinct ways. We can either require
\be
|f(s)|<|\re b|/2,\ \ \ \forall s\in\R,
\ee
 or we can restrict the slope of $\cC$:
\be
|f'(s)|<|\re b|/|\im b|,\ \ \ \forall s\in\R. 
\ee

Next, the contour $\cC$ should be further restricted to ensure that the $J_2$-factor in the integrand~\eqref{I3} remains holomorphic when $z_1$ and $z_2$ vary over $\cC$. For $b\in (0,2a)$ this can be readily achieved, since in that case $J_2$ is holomorphic for $\re z_1 \ne \re z_2$. For~$\im b\ne 0$, however, we need $|\im (z_1-z_2)|<2a-\re b$ whenever $\re (z_1-z_2)$ equals $\pm \im  b$. Again, we can ensure this  by requiring either 
\be
|f(s)|<|2a-\re b|/2,\ \ \ \forall s\in\R,
\ee
 or 
\be
|f'(s)|<|2a-\re b|/|\im b|,\ \ \ \forall s\in\R. 
\ee

Accordingly, in the following analog of~Prop.~4.5 we obtain an extended holomorphy domain for~$J_3$ by allowing only contours for which either
\be\label{fres}
|f(s)|<m(\re b)/2,\ \ \ \forall s\in\R,
\ee
or
\be\label{fpres}
|f'(s)|<m(\re b)/|\im b|,\ \ \ \forall s\in\R,
\ee
where $m(d)$ is given by
\be\label{md}
m(d)\equiv \min (2a-d,d),\ \ d\in (0,2a).
\ee

Specifically, we define two domains
\begin{multline}\label{D3Im}
D_3^I\equiv \{ ( b,x)\in S_a\times\C^3 \mid  |\re (x_m- x_n)|\ne |\im b|, 
\\
\max_{1\le j<k\le 3}|\im (x_j- x_k)|<2a-\re b +m(\re b)/2\},
\end{multline}
\begin{multline}\label{D3Re}
D_3^R    \equiv   \{ ( b,x)\in S_a\times\C^3 \mid  |\re (x_m-x_n)|>|\im b|>0, 
 \\
 |\im(x_m-x_n)\,\im b|<(|\re (x_m-x_n)|-|\im b|)m(\re b)\},
\end{multline}
where $1\le m<n\le 3$ and $k=1,2,3$. 

\begin{proposition}
Let $y\in\R^3$. Then the function~$J_3(b;x,y)$ is holomorphic in  $(b,x)$ on the domain
\be
D_3^e \equiv  D_3\cup D_3^I\cup D_3^R.
\ee
\end{proposition}
\begin{proof}
 First assume $(b,x)\in D_3^I$. We start from the representation~\eqref{J3eta}. Then the first pole in the upward/downward sequences~\eqref{I3po} has an imaginary part that is below/above the contour~$\R+i\phi(x)$ by a distance less than $m(\re b)/2$, so we need only use downward/upward triangular indentations  whose heights are smaller than~$m(\re b)/2$. (Note that we can make disjoint indentations, since $|\re (x_m- x_n)|\ne |\im b|$.) This ensures that~$z_1-z_2$ stays away from the $W_2$-poles and the $J_2$-argument stays in~$D_2$. Hence holomorphy in~$D_3^I$ results.

Next, let $(b,x)\in D_3^R$. The restriction on the real parts implies that the numbers $\re x_m$, $m=1,2,3$, are distinct. By permutation symmetry, we need only define the piecewise linear contour $\cC$ for the case~$\re x_3<\re x_2<\re x_1$. Then we have
\be
\re x_3+|\im b|/2<\re x_2-|\im b|/2,\ \ \re x_2+|\im b|/2<\re x_1-|\im b|/2.
\ee
First, we connect the pairs of points 
\be
P_m^{\pm}\equiv x_m\pm |\im b|/2, \ \ m=1,2,3,
\ee
with (horizontal) line segments. Then we connect~$P_3^{+}$, $P_2^{-}$ and $P_2^{+}$, $P_1^{-}$ with line segments. The restrictions on~$D_3^R$ imply that the slopes of the latter satisfy~\eqref{fpres}. Finally, we connect $P_3^{-}$ to~$(-R,0)$ and $P_1^{+}$ to~$(R,0)$, where $R>0$ is chosen sufficiently large so that the slopes~$s_{\pm}$ of the two line segments also satisfy~$|s_{\pm}|<m(\re b)/|\im b|$. Completing $\cC$ with the intervals~$(-\infty,-R)$ and~$(R,\infty)$, we stay in the region where the integrand is holomorphic. Hence holomorphy of~$J_3$ in~$D_3^{R}$ follows as well.
\end{proof}


\section{The inductive step}

In this section we use the validity of Props.~5.1--5.4 as the starting point for the inductive step $N-1\to N$. More specifically,  Props.~6.1--6.4 below reduce to the former for $N=3$, and our induction assumption is that they hold true if we replace~$N$ by~$N-1$. This has the advantage of expository conciseness, but the minor drawback that we need to invoke Prop.~6.3 with $N\to N-1$ to obtain Prop.~6.1. As it turns out, our account in Section~5 can be followed to a large extent. Accordingly, we sometimes skip details when their general-$N$ version will be clear from Section~5.  

The integrand arising when constructing $J_N$ from $J_{N-1}$ is given by (cf.~\eqref{JN}) 
\be\label{IN1}
I_{N}(b;x,y,z)=W_{N-1}(b;z)\cS^{\sharp}_N(b;x,z)J_{N-1}(b;z,(y_1-y_N,\ldots,y_{N-1}-y_N)).
\ee
From \eqref{rbc} we obtain the bound
\be\label{SNb}
|\cS^{\sharp}_N(b;x,z)|<\prod_{j=1}^N\prod_{k=1}^{N-1}\frac{C_r(b,z_k-x_j)}{\cosh(\gamma\re(z_k-x_j))},\ \ \  \gamma=\alpha\re b/2,
\ee
as long as $(b,z_k-x_j)\in\cD_r$ for $j=1,\ldots,N$ and $k=1,\ldots,N-1$. Also, from~\eqref{wbr} we have
\be\label{WN-1be}
|W_{N-1}(b;z)|<\prod_{1\leq j<k\leq N-1}C_w(b)\sinh^2(\gamma(z_j-z_k)),\ \ \ \im z_1=\cdots=\im z_{N-1}.
\ee

In order to estimate the factor $J_{N-1}$ in~\eqref{IN1}, we assume that the imaginary parts of~$z_1,\ldots,z_{N-1}$ are equal to some $c\in\R$, and replace~$N$ by~$N-1$ in Prop.~6.3 below. Then the difference function $d_{N-1}(z)$ (given by~\eqref{d}  and~\eqref{imordN}) vanishes and $\phi_{N-1}(z)$~\eqref{phi} reduces to $c$. Hence we obtain the bound
\begin{multline}\label{JN-1bea}
|J_{N-1}(b;z,(y_1-y_N,\ldots,y_{N-1}-y_N))|\\ < C(b)\exp\left(-\alpha c\sum_{j=1}^{N-1}(y_j-y_N)\right)\prod_{1\le m<n\le N-1}\frac{z_m-z_n}{\sinh(\gamma(z_m-z_n))},
\end{multline}
where $C$ is continuous on $S_a$. Together with the bounds \eqref{SNb} and \eqref{WN-1be}, this implies that $I_N$ has exponential decay when one or more of the quantities $|\re z_1|,\ldots,|\re z_{N-1}|$ diverge, uniformly for $c$ varying over bounded intervals.

We note that the poles of the kernel function $\cS^{\sharp}_N(b;x,z)$ are located at (cf.~\eqref{Gpo}--\eqref{Gze})
\be\label{INpo}
z_k=x_j\pm(ib/2-ia-z_{lm}),\ \ k=1,\ldots,N-1,\ \ j=1,\ldots,N,\ \ \ l,m\in\N.
\ee
In addition, by the induction assumption, we may make use of Prop.~6.1 after taking $N\to N-1$, and hence conclude that $J_{N-1}(b;z,(y_1-y_N,\ldots,y_{N-1}-y_N))$ is holomorphic in
$D_{N-1}$.
Requiring at first $x\in\R^N$, it now follows that the integrand $I_N$ is regular for $z\in\R^{N-1}$, so that
\be\label{JNdef}
J_N(b;x,y)\equiv \frac{e^{i\alpha y_N(x_1+\cdots+x_N)}}{(N-1)!}\int_{\R^{N-1}}dzI_N(b;x,y,z),\ \ (b,x,y)\in S_a\times\R^N\times\R^N,
\ee
is well defined.

It is clear that $J_N(b;x,y)$ (as initially defined by \eqref{JNdef}) extends to a function that is holomorphic in $b$ and $x_j$ for $b\in S_a$ and $|\im x_j|<a-\re b/2$, $j=1,\ldots,N$. Also, the above bounds imply that we are allowed to shift all contours $\R$ up and down by the same amount, provided none of the poles \eqref{INpo} are met. In this way we can recursively extend the holomorphy domain. In order to make this precise, we set 
\be\label{phi}
\phi_N(x)\equiv\im(x_{j_1}+x_{j_N})/2,\ \ \ \ x\in\C^N,
\ee
where the indices $j_1$ and $j_N$ are such that
\be\label{imordN}
\im x_{j_{N}}\le \im x_{j_{N-1}}\le\cdots\le \im x_{j_2}\le \im x_{j_1},\ \ \ \{j_{1},\ldots,j_N\}=\{1,\ldots,N\}.
\ee
Then we have the following general-$N$ version of Prop.~5.1.

\begin{proposition}
Let $y\in\R^N$. Then the function $J_N(b;x,y)$ is holomorphic in
\be\label{DN}
D_N\equiv \{(b,x)\in S_a\times\mathbb{C}^N \mid \max_{1\leq j<k\leq N}|\im (x_j-x_k)|<2a-\re b\}.
\ee
\end{proposition}

\begin{proof}
Let us fix $(b,x)\in D_N$. Then we have
\be
\im x_j+a-\re b/2>\phi_N(x),\ \ \im x_j-a+\re b/2<\phi_N(x),\ \ \ j=1,\ldots,N.
\ee
It follows that the contour $\R+i\phi_N(x)$ stays below/above the upward/downward $G$-pole sequences \eqref{INpo}. Thus we can continue $J_N$ to the $x$-value in question by simultaneously shifting all contours, without passing any of these poles. In this way we obtain the representation
\be\label{JNeta}
J_N(b;x,y)=\frac{e^{i\alpha y_N(x_1+\cdots+x_N)}}{(N-1)!}\int_{(\R+i\phi_N(x))^{N-1}}dzI_N(b;x,y,z),\ \ (b,x)\in D_N,\ \ y\in\R^N,
\ee
whence holomorphy in $D_N$ is plain.
\end{proof}

In our next proposition we prove that $J_N$ is a joint eigenfunction of the $2N$ A$\De$Os~\eqref{Aks}.   
 As in Prop.~5.2, we need to shift some of the~$N-1$ contours in the proof. To show that this does not cause problems with the contour tails, we can no longer rely on the bounds \eqref{WN-1be} and \eqref{JN-1bea}, since they only hold for $\im z_1=\cdots=\im z_{N-1}$.

However, it follows from the bounds \eqref{wbc}--\eqref{Cwas} that the rate at which $W_{N-1}(z)$ diverges as $|\re(z_m-z_n)|\to\infty$, $1\leq m<n\leq N-1$, is independent of $\im(z_m-z_n)$. Moreover, by the induction assumption, we may invoke Prop.~6.3 for $N\to N-1$. As long as $|\im(z_m-z_n)|<2a-\re b$, this entails that on closed subintervals the decay rate as $|\re(z_m-z_n)|\to\infty$ in the bound on the $J_{N-1}$-factor in the integrand $I_N$ \eqref{IN1} does not depend on $\im(z_m-z_n)$ either. Combined with the bound \eqref{SNb} on the kernel function~$\cS_N^{\sharp}$, these arguments show that the contour shifts are legitimate.

\begin{proposition}
Fixing $y\in\R^N$ and letting $(b,x)$ vary over
\be
D_{N,s}\equiv\{(b,x)\in D_N \mid (b,x-i\eta)\in D_N, \forall\eta\in [0,a_l]^N,\ \re b\in (0,a_s)\}, 
\ee
  we have the eigenvalue equations
\be\label{ANeig}
A^{(N)}_{k,\delta}(x) J_N(x,y)=S^{(N)}_k(e_{\de}(2y_1),\ldots,e_{\de}(2y_N))J_N(x,y),
\ee
where $k=1,\ldots,N$, and $\de=+,-$.
\end{proposition}

\begin{proof}
The six steps in the proof of Prop.~5.2 can be readily adapted, so we only indicate the changes. The first two steps yield
\begin{multline}
e_{\de}(2ky_N)\frac{e^{i\alpha y_N(x_1+\cdots+x_N)}}{(N-1)!}\\ \times\int_{C_-}dz_1\cdots\int_{C_-}dz_{N-1}W_{N-1}(z)A^{(N)}_{k,\delta}(x)\cS^{\sharp}_N(x,z)J_{N-1}(z,(y_1-y_N,\ldots,y_{N-1}-y_N)),
\end{multline}
as the generalization of~\eqref{3lhs1}.

For the third step, we invoke the kernel identity \eqref{key} and the symmetric function recurrence \eqref{Srec} to deduce that we need only show that the integrals
\be\label{Nlhs}
\int_{C_-}dz_1\cdots\int_{C_-}dz_{N-1}W_{N-1}(z)J_{N-1}(z,(y_1-y_N,\ldots,y_{N-1}-y_N))A^{(N-1)}_{l,\delta}(-z)\cS^{\sharp}_N(x,z),
\ee
are equal to
\begin{multline}\label{eigN-1l}
S^{(N-1)}_l\big(e_{\de}(2(y_1-y_N)),\ldots,e_{\de}(2(y_{N-1}-y_N))\big)\\ \times(N-1)!\exp(-i\alpha y_N(x_1+\cdots+x_N))J_N(x,y),
\end{multline}
for $l=1,\ldots,N-1$.

Fourth, we fix $l$ and use \eqref{Aks} to write \eqref{Nlhs} as 
\begin{multline}\label{Nlhs1}
\int_{C_-}dz_1\cdots\int_{C_-}dz_{N-1}W_{N-1}(z)J_{N-1}(z,(y_1-y_N,\ldots,y_{N-1}-y_N))\\ \times \sum_{\substack{I\subset\lbrace 1,\ldots,N-1\rbrace\\ |I|=l}}\prod_{\substack{m\in I\\ n\notin I}}\frac{s_{\delta}(z_m-z_n+ib)}{s_{\delta}(z_m-z_n)}\prod_{m\in I}\exp(ia_{-\delta}\partial_{z_m})\cS^{\sharp}_N(x,z).
\end{multline}
Changing variables $z_m\to z_m-ia_{-\de}$, $m\in I$, we rewrite \eqref{Nlhs1} as
\begin{multline}
\sum_{\substack{I\subset\{ 1,\ldots,N-1\}\\ |I|=l}}\left(\prod_{m\in I}\int_{C_-+ia_{-\de}}dz_m\right)\left(\prod_{n\notin I}\int_{C_-}dz_n\right) \cS^{\sharp}_N(x,z)\prod_{\substack{m\in I\\ n\notin I}}\frac{s_{\delta}(z_m-ia_{-\de}-z_n+ib)}{s_{\delta}(z_m-ia_{-\de}-z_n)}\\ \times\prod_{m\in I}\exp(-ia_{-\delta}\partial_{z_m})W_{N-1}(z)J_{N-1}(z,(y_1-y_N,\ldots,y_{N-1}-y_N)),
\end{multline}
and then use \eqref{Wade}  to see this equals
\begin{multline}
\sum_{\substack{I\subset\{ 1,\ldots,N-1\}\\ |I|=l}}\left(\prod_{m\in I}\int_{C_-+ia_{-\de}}dz_m\right)\left(\prod_{n\notin I}\int_{C_-}dz_n\right)W_{N-1}(z)\cS^{\sharp}_N(x,z)\\ \times \prod_{\substack{m\in I\\ n\notin I}}\frac{s_{\de}(z_m-z_n-ib)}{s_{\de}(z_m-z_n)}\prod_{m\in I}\exp(-ia_{-\de}\partial_{z_m})J_{N-1}(z,(y_1-y_N,\ldots,y_{N-1}-y_N)).
\end{multline}

We claim that when we simultaneously shift the contours $C_-$ up by $a_{-\de}$ for all variables $z_n$, $n\notin I$,   we do not meet poles. (The paragraph preceding the proposition implies that these shifts do not create problems at the tail ends.) Assuming this claim is valid, we note that the integrands are regular for $z_1,\ldots,z_{N-1}\in C_-+ia_{-\de}$, and since the shifts ensure that the $N-1$ contours are equal in all integrals, we may use the eigenvalue equation \eqref{ANeig} with $N\to N-1$ and $k\to l$ (the fifth step). In the sixth step, we shift the contours back to $\R$, hence arriving at \eqref{eigN-1l}.

It remains to prove the claim. Obviously, we stay in the holomorphy domain of~$\cS^{\sharp}_N$ while shifting the contours. Using Prop.~6.1 with $N\to N-1$, we see this is also true for $J_{N-1}$. Now we fix $I\subset\{1,\ldots,N-1\}$ such that $|I|=l$ and first choose $a_{-\de}=a_s$. When shifting the contours for $z_n$, $n\notin I$, we only meet the simple poles of the factors $1/G(z_m-z_n+ia-ib)$ with $m\in I$ and $n\notin I$ at $z_m-z_n=ib$, but they are matched by zeros of the factors $s_{\de}(z_m-z_n-ib)$. Moreover, the simple poles arising from $s_{\de}(z_m-z_n)$ at $z_m-z_n=0$ and $z_m-z_n=ia_l$ (encountered when $a_s=a_l$) are matched by zeros of the factors $G(z_m-z_n+ia)$ in $W_{N-1}(z)$. It follows that the integrands are regular for $z_m\in C_-+ia_s$, $m\in I$, and $z_n$ in the strip $\im z_n\in [-a_l/2,-a_l/2+a_s]$, $n\notin I$, as claimed.

Choosing next $a_{-\de}=a_l$, the simple poles of $1/G(z_m-z_n+ia-ib)$, $m\in I,n\notin I$, located at $z_m-z_n=ib+ika_s$ for $ka_s\leq a_l-b$, are matched by zeros of $s_{\de}(z_m-z_n-ib)$, and the simple poles resulting from $s_{\de}(z_m-z_n)$ at $z_m-z_n=ika_s$ for $ka_s\leq a_l$ are matched by zeros of $G(z_m-z_n+ia)$. Hence the integrands are also regular for $z_m\in C_-+ia_l$, $m\in I$, and $\im z_n\in [-a_l/2,a_l/2]$, $n\notin I$, as claimed.
\end{proof}

The generalization of Prop.~5.3 involves the function~$\phi_N(x)$~\eqref{phi} and the distance function
\be\label{d}
d_N(x)\equiv \im(x_{j_1}-x_{j_N}),\ \ \ x\in\C^N,
\ee
where $j_1$ and $j_N$ are given by \eqref{imordN}.

\begin{proposition}
Let $y\in\mathbb{R}^N$. For any $(b,x)\in D_N$ we have
\bea\label{JNb1}
|J_N(b;x,y)| & < & C(b,d_N(x))\exp\left(-\alpha \left[\phi_N(x)\sum_{j=1}^Ny_j+y_N\sum_{k=2}^{N-1}\big(\im~x_{j_k}-\phi_N(x)\big)\right]\right)
\nonumber \\
&  &  \times \prod_{1\le m<n\le N}\frac{ \re(x_m-x_n)}{\sinh(\gamma\,\re(x_m-x_n))},
\eea
where $C$ is continuous on $S_a\times [0,2a-\re b)$.
\end{proposition}

\begin{proof}
From the representation \eqref{JNeta} we have the majorization
\begin{multline}
|J_N(x,y)| \leq \frac{\exp\left(-\alpha y_N\sum_{j=1}^N\im x_j\right)}{(N-1)!}\\ \times \int_{(\mathbb{R}+i\phi_N(x))^{N-1}}dz|W_{N-1}(z)\cS^{\sharp}_N(x,z)J_{N-1}(z,(y_1-y_N,\ldots,y_{N-1}-y_N))|.
\end{multline}
By the induction assumption, \eqref{JNb1} holds true with $N$ replaced by $N-1$. Applying this to the factor~$J_{N-1}$ of the integrand, we observe that~$d_{N-1}(z)=0$ and~$\phi_{N-1}(z)=\phi_N(x)$. Using also the bounds \eqref{SNb} and \eqref{WN-1be}, we now infer
\begin{multline}
	|J_N(b;x,y)|\leq C(b)\exp\left(-\alpha \left[\phi_N(x)\sum_{j=1}^Ny_j+y_N\sum_{k=2}^{N-1}\big(\im~x_{j_k}-\phi_N(x)\big)\right]\right)\\ \times\int_{\mathbb{R}^{N-1}}ds\prod_{1\leq m<n\leq N-1}(s_m-s_n)\sinh(\gamma(s_m-s_n))\prod_{j=1}^N\prod_{k=1}^{N-1}\frac{C_r(b;s_k+i\phi_N(x)-x_j)}{\cosh(\gamma(s_k-\re x_j))}.
\end{multline}
Just as in the proof of Prop.~5.3 (cf.~\eqref{phid}), we see that the $C_r$-product is bounded above by a function $C(b,d_N(x))$ that is continuous on $S_a\times [0,2a-\re b)$. Using Lemma~C.2 to compute the remaining integral, \eqref{JNb1} results.
\end{proof}

  As in the $N=3$ case, we can allow complex~$y$, provided we restrict attention to a subdomain of $D_N$ for the dependence on $(b,x)$. First, we introduce general-$N$ notation that deviates slightly from the notation~\eqref{not3} for the case $N=3$: For a given vector $v\in\C^M$, $M>1$, we define $V_M\in\C$ and $v^{(M)}\in\C^M$ by
  \be\label{notM}
V_M\equiv \frac1M\sum_{j=1}^Mv_j,\ \ \ v^{(M)}_j\equiv v_j-V_M,\ \ j=1,\ldots,M.
\ee
 Then we need to work with
\be
D_N^r\equiv\{(b,x)\in S_a\times\C^N \mid \big|\im x_j^{(N)}\big|<a-\re b/2,\ j=1,\ldots,N\}\subset D_N,
\ee
 for the representation
\be\label{JNX}
J_N(x,y)=\frac{e^{i\alpha y_N(x_1+\cdots+x_N)}}{(N-1)!}\int_{(\R+i\im X_N)^{N-1}}dz I_N(x,y,z),\ \ \ \ y\in\R^N,
\ee
to hold true.

We now use \eqref{JNX} as a starting point for concluding that~$J_N(x,y)$ equals the product of a center-of-mass factor~$\exp(Ni\alpha X_NY_N)$ and a function $J_N^r(x,y)$ depending only on the differences~$x_j-x_{j+1}, y_j-y_{j+1}$, $j=1,\ldots,N-1$. The resulting representation is important in its own right, and it will enable us to control the $y$-dependence. Specifically,  we claim that we have
\be\label{JNs}
J_N(x,y)=\exp(Ni\alpha X_NY_N)J_N^r(x,y),
\ee
\be\label{JNr}
J_N^r(x,y)\equiv \frac{1}{(N-1)!}\int_{\R^{N-1}}ds W_{N-1}(s)\cS^{\sharp}_N(x^{(N)},s)J_{N-1}(s,(y_1-y_N,\ldots,y_{N-1}-y_N)),
\ee
where $(b,x)\in D_N^r$ and $y\in \R^N$. 

As before, our proof of this claim proceeds by induction on~$N$, noting~\eqref{JNs}--\eqref{JNr} reduce to~\eqref{J3s}--\eqref{J3r} for~$N=3$. First, we observe that the integral on the rhs of \eqref{JNX} can be rewritten as
\be
\int_{\R^{N-1}}dsW_{N-1}(s)\cS^{\sharp}_N(x^{(N)},s)J_{N-1}((s_1+X_N,\ldots,s_{N-1}+X_N),(y_1-y_N,\ldots,y_{N-1}-y_N)).
\ee
Next, we take $N\to N-1$ in \eqref{JNs}--\eqref{JNr}, to deduce from the induction assumption
\begin{multline}
J_{N-1}((s_1+X_N,\ldots,s_{N-1}+X_N),(y_1-y_N,\ldots,y_{N-1}-y_N))\\ = \exp\left(i\alpha X_N\sum_{j=1}^{N-1}(y_j-y_N)\right) J_{N-1}(s,(y_1-y_N,\ldots,y_{N-1}-y_N)).
\end{multline}
Using the identity
\be
y_N\sum_{j=1}^Nx_j+X_N\sum_{j=1}^{N-1}(y_j-y_N)=NX_NY_N,
\ee
 the claim now readily follows.

Setting
\be
\cD_N\equiv \{(b,x,y)\in D_N^r\times\C^N \mid |\im(y_j-y_k)|<\re b,\ j,k=1,\ldots,N\},
\ee
\be
\mu_N(x)\equiv\max_{j=1,\ldots,N}|\im x^{(N)}_j|,
\ee
we are in the position to state and prove the arbitrary-$N$ version of Prop.~5.4.

\begin{proposition}
The function $J_N(b;x,y)$ is holomorphic in $(b,x,y)$ on $\cD_N$. For any $(b,x,y)\in\cD_N$ with
\be\label{xyres2}
\re(x_j-x_k)\ne 0,\ \ \ \im(y_j-y_k)\ne 0,\ \ \ 1\leq j<k\leq N,
\ee
we have
\begin{multline}\label{JNby}
|J_N(b;x,y)| < C(b,\mu_N(x))\exp(-N\alpha\im(X_NY_N))\\ \times\frac{\sum_{\sigma\in S_N}(-)^{\sigma}\exp\left(-\alpha\re\big(x^{(N)}\big)\cdot\sigma\big(\im y^{(N)}\big)\right)}{\prod_{1\leq m<n\leq N}\sin(\pi\im(y_n-y_m)/\re b)\sinh(\gamma\re(x_m-x_n))},
\end{multline}
where $C$ is continuous on $S_a\times [0,a-\re b/2)$.
\end{proposition}

\begin{proof}
It is clear from \eqref{JNs} that we only have to estimate $J_N^r$. The induction assumption amounts to \eqref{JNby} holding true when we replace $N$ by~$N-1$. Combining the resulting bound on $J_{N-1}$ with \eqref{SNb} and \eqref{WN-1be}, we obtain via a straightforward computation (using the general-$N$ counterpart of~\eqref{sums3}) 
\begin{multline}\label{JNrb}
|J_N^r(x,y)| < \frac{C(b,\mu_N(x))}{(N-1)!\prod_{1\leq m<n\leq N-1}\sin(\pi\im(y_n-y_m)/\re b)}\\ \int_{\R^{N-1}}ds\frac{\prod_{1\leq j<k\leq N-1}\sinh(\gamma(s_j-s_k))}{\prod_{j=1}^N\prod_{k=1}^{N-1}\cosh(\gamma(s_k-\re  x^{(N)}_j))} \exp\left(\alpha S_{N-1}\sum_{j=1}^{N-1}\im(y_N-y_j)\right)\\ \times \sum_{\tau\in \cS_{N-1}}(-)^{\tau}\exp\left(\alpha\sum_{j=1}^{N-2}s^{(N-1)}_{\tau(j)}\im(y_{N-1}-y_j)\right).
\end{multline}

Next, we introduce
\begin{multline}\label{AN-1}
A_{N-1}(t,u)\equiv \int_{\R^{N-1}}ds\frac{\prod_{1\leq j<k\leq N-1}\sinh(\gamma(s_j-s_k))}{\prod_{j=1}^N\prod_{k=1}^{N-1}\cosh(\gamma s_k-t_j)}\exp\left(\gamma S_{N-1}\sum_{j=1}^{N-1}(u_N-u_j)\right)\\ \times\sum_{\tau\in S_{N-1}}(-)^{\tau}\exp\left(\gamma\sum_{j=1}^{N-2}s^{(N-1)}_{\tau(j)}(u_{N-1}-u_j)\right).
\end{multline}
It is plain that $A_{N-1}(\gamma\re x^{(N)},2\im y/\re b)$ equals the integral in the bound \eqref{JNrb}. Now we fix $\tau\in S_{N-1}$ and consider the corresponding summand on the rhs of \eqref{AN-1}. Changing variables $s_j\to s_{\tau^{-1}(j)}$ and using antisymmetry of the $\sinh$-product, we obtain the summand with $\tau={\rm id}$. Hence we may replace the sum by $(N-1)!$ times this summand. We simplify the resulting product of exponential functions by using the identity
\be
S_{N-1}\sum_{j=1}^{N-1}(u_N-u_j) + \sum_{j=1}^{N-2}s^{(N-1)}_j(u_{N-1}-u_j) = \sum_{j=1}^{N-1}s_j(u_N-u_j),
\ee
which is readily verified. Taking $s\to z/\gamma$ and comparing the result to~\eqref{Cintegral}, we deduce
\be
A_{N-1}(t,u)=\frac{(N-1)!}{\gamma^{N-1}}C_{N-1}(u,t).
\ee
The upshot is that the bound \eqref{JNrb} becomes
\begin{equation}\label{JNC2}
|J_N^r(x,y)| < \frac{C(b,\mu_N(x))}{\gamma^{N-1}\prod_{1\leq m<n\leq N-1}\sin(\pi\im(y_n-y_m)/\re b)}C_{N-1}(2\im y/\re b,\gamma\re x^{(N)}).
\end{equation}
Using again the general-$N$ version of~\eqref{sums3},
 the assertions now follow from~Lemma~C.3.
\end{proof}

We point out that the bound~\eqref{JNby} is symmetric in~$y$, a feature we also expect for~$J_N(x,y)$, cf. the next section. Another point of interest is its specialization to real~$y$. Using~\eqref{FN}--\eqref{FNpow} we see that this is of the form
\be\label{JNbnew}
|J_N(b;x,y)| < C(b,\mu_N(x))\exp(-N\alpha Y_N\im X_N)\prod_{1\le m<n\le N}\frac{ \re(x_m-x_n)}{\sinh(\gamma\,\re(x_m-x_n))},\ \ \ y\in\R^N.
\ee
It should be noted that for $N>2$ this does not coincide with the bound~\eqref{JNb1} obtained in~Prop.~6.3. Moreover, the latter holds true on~$D_N$, whereas we need~$(b,x)\in D_N^r$ to arrive at~\eqref{JNbnew}.

We conclude this section by pointing out that since~$|J_N(b;x,y)| $ is manifestly non-negative, it follows from~\eqref{JNby} that the function on the rhs is positive. This corollary of~Prop.~6.4 is by no means obvious, even for~$N=3$. Rephrased in terms of the function~$Q_{N-1}(u,t)$ given by~\eqref{QN}, it implies
\be\label{QNpos}
Q_{N-1}(u,t)>0,\ \ \ |u_j-u_k|<2,\ \ 1\le j<k\le N,\ \ \ u,t\in\R^N.
\ee


\section{Outlook}

In this paper we have focused on the first steps in the recursive scheme for constructing joint eigenfunctions~$J_N$  of the A$\De$Os $A_{k,\de}$ (given by~\eqref{JN} and \eqref{Aks}, resp.). More specifically, we have concentrated on establishing holomorphy domains and uniform decay bounds that were sufficient for proving that the scheme does provide well-defined functions~$J_N$ that satisfy the expected joint eigenvalue equations. There are numerous aspects  that remain to be investigated. Below we briefly discuss a few that we plan to study in future articles, associated with conjectured features of the joint eigenfunctions. In the free cases $b=a_{\pm}$ these features all follow from the results in Section~3, whereas for $N=2$   they can be gleaned from \cite{Rui11}. 

$\mathbf{S_N\times S_N}$ {\bf invariance:} We expect that $J_N(x,y)$ is invariant under permutations of both the variables $x\equiv(x_1,\ldots,x_N)$ and the variables $y\equiv(y_1,\ldots,y_N)$. We add that the former symmetry property is immediate from the defining formula \eqref{JN}, whereas the latter is far from clear. We also recall that the bound~\eqref{JNby} exhibits the expected symmetry.

{\bf Self-duality:} In the free cases $b=a_{\de}$, $\de=\pm$, it is plain from Theorem 3.1 that upon multiplication of $J_N(a_{\de};x,y)$ by the factor $\prod_{1\leq j<k\leq N}s_{\de}(y_j-y_k)/s_{-\de}(y_j-y_k)$ we obtain a function that is symmetric under the interchange of the variables $x$ and $y$. More generally, for the function
\be
\cJ_N(b;x,y)\equiv J_N(b;x,y)\prod_{1\leq j<k\leq N}\, \prod_{\de=+,-}G(\de(y_j-y_k)+ia-ib) ,
\ee
we conjecture the self-duality property
\be
\cJ_N(b;x,y)=\cJ_N(b;y,x).
\ee

{\bf Parameter symmetries:} Since the hyperbolic gamma function $G(a_+,a_-;z)$ is invariant under the interchange $a_+\leftrightarrow a_-$, it is clear from the defining formula \eqref{JN} that $J_N(a_+,a_-,b;x,y)$ has this parameter symmetry, too. Moreover, as we noted in the introduction, the Hamiltonians $H_{k,\de}$ are invariant under the substitution $b\to 2a-b$. We conjecture that their joint eigenfunctions
\be
\cF_N(b;x,y)\equiv G(ib-ia)^{N(N-1)/2}W_N(b;x)^{1/2}\cJ_N(b;x,y)W_N(b;y)^{1/2},
\ee
are also invariant under this substitution:
\be
\cF_N(2a-b;x,y) = \cF_N(b;x,y).
\ee

{\bf Global meromorphy:} Although we have only considered holomorphy properties in this paper, we do expect that the recursive scheme yields globally meromorphic joint eigenfunctions. Moreover, the free cases and the $N=2$ case suggest both the location of the poles and upper bounds on their orders. Specifically, we expect that the product function
\be
\cP_N(b;x,y)\equiv \cJ_N(b;x,y)\prod_{1\leq j<k\leq N}\prod_{\de=+,-}E(\de(x_j-x_k)+ib-ia)E(\de(y_j-y_k)+ib-ia),
\ee
has a jointly analytic extension to all $(b,x,y)\in S_a\times\C^N\times\C^N$. (The $E$-function is an entire function related to the hyperbolic gamma function by $G(z)=E(z)/E(-z)$, from which the location and order of the $E$-zeros can be read off, cf.~\cite{Rui99}.)

{\bf Soliton scattering:} It is a long-standing conjecture that the particles in the relativistic Calogero-Moser systems of hyperbolic type exhibit soliton scattering (conservation of momenta and factorization). In more detail, provided $b\in(0,2a)$, we expect 
\bea
\cF_N(b;x,y) & \sim & C_N\sum_{\sigma\in S_N}\prod_{\substack{1\leq j<k\leq N\\ \sigma^{-1}(j)<\sigma^{-1}(k)}}(-u(y_j-y_k))^{1/2} \nonumber \\ & & \times \prod_{\substack{1\leq j<k\leq N\\ \sigma^{-1}(j)>\sigma^{-1}(k)}}(-u(y_j-y_k))^{-1/2}\exp(i\alpha x\cdot\sigma(y)),
\eea
for $x_j-x_{j+1}\to\infty$, $j=1,\ldots,N-1$, with $C_N$ a constant and the scattering function $u(z)$ given by \eqref{uU}.

{\bf Orthogonality and completeness:} We conjecture that for all $b\in[0,2a]$, the function $\cN_N\cF_N(b;x,y)$ (with~$\cN_N$ a normalization constant) yields the kernel of a unitary operator on $L^2(F_N,dx)$, where~$F_N$ is the configuration space
\be
F_N\equiv\{-\infty<x_N<\cdots<x_1<\infty\}.
\ee
In the free cases this unitary operator amounts to the antisymmetric version of Fourier transformation, cf.~Section~3.


\renewcommand{\thesection}{A}
\setcounter{equation}{0}

\addcontentsline{toc}{section}{Appendix A. The hyperbolic gamma function}
																																												
\section*{Appendix A. The hyperbolic gamma function}
In this appendix we collect properties of the hyperbolic gamma function $G(a_+,a_-;z)$ we have occasion to use. Their proofs can be found in~Subsection~III~A of~\cite{Rui97}, but for the $G$-asymptotics we need the sharper bounds obtained in~Theorem~A.1 of~\cite{Rui99}.  We fix
\be
a_+,a_- >0,
\ee
and suppress the dependence of $G$ on $a_+,a_-$ whenever this will not cause ambiguities. 

 The function $G(z)$ can be defined as the unique minimal solution of one of the two analytic difference equations
\be\label{Gades}
	\frac{G(z+ia_\delta/2)}{G(z-ia_\delta/2)} = 2c_{-\delta}(z),\ \ \ \  \delta=+,-,
\ee
that has modulus 1 for real $z$ and satisfies $G(0)=1$ (cf.~\eqref{sce} for the notation used here); it is not obvious, but true that this entails that the other one is then satisfied as well.
It is meromorphic in~$z$, 
and for $z$ in the strip 
\be\label{strip}
S\equiv \{z\in\C \mid |\im (z)|<a\},\ \ \ a=(a_++a_-)/2,
\ee
no poles and zeros occur. Hence we have
\be\label{Gg}
G(z)=\exp(ig(z)),\ \ \ \ z\in S,
\ee
with $g(z)$ holomorphic in $S$. 
 Explicitly, $g(z)$ has the integral representation
\be\label{ghyp}
g(a_{+},a_{-};z) =\int_0^\infty\frac{dy}{y}\left(\frac{\sin 2yz}{2\sinh(a_{+}y)\sinh(a_{-}y)} - \frac{z}{a_{+}a_{-} y}\right),\ \ \ \ z\in S.
\ee
From this, the following properties of the hyperbolic gamma function are immediate:
\be\label{refl}	
G(-z) = 1/G(z),\ \ \ ({\rm reflection\ equation}),
\ee
\be\label{modinv}
G(a_-,a_+;z) = G(a_+,a_-;z),\ \ \  ({\rm modular\ invariance}),
\ee
\be\label{sc}
	G(\lambda a_+,\lambda a_-;\lambda z) = G(a_+,a_-;z),\quad \lambda\in(0,\infty),\ \ \ ({\rm scale\ invariance}),
\ee
\be\label{Gcon}
\overline{G(a_+,a_-;z)}=G(a_+,a_-;-\overline{z}).
\ee

Defining
\be\label{zkl}
z_{kl}\equiv ika_{+}+ila_{-},\ \ \ \ k,l\in\N\equiv \{ 0,1,2,\ldots\},
\ee
the hyperbolic gamma function has its poles at 
\be\label{Gpo}
z=z_{kl}^-,\ \ \ z_{kl}^-\equiv -ia -z_{kl},\ \ \ \ k,l\in\N,\ \ \ (G{\rm -poles}),
\ee
and its zeros at
\be\label{Gze}
z=z_{kl}^+,\ \ \ z_{kl}^+\equiv ia +z_{kl},\ \ \ \ k,l\in\N,\ \ \ \ (G{\rm -zeros}).
\ee
The pole at $-ia$ is simple and has residue
\be\label{Gres}
\lim_{z\to -ia}(z+ia)G(z)=\frac{i}{2\pi}(a_{+}a_{-})^{1/2}.
\ee
More generally, for a given $(k_0,l_0)\in\N^2$, the multiplicities of the pole~$z_{k_0l_0}^{-}$ and zero~$z_{k_0l_0}^{+}$ are equal to the number of distinct pairs~$(k,l)\in\N^2$ such that ~$z_{kl}^{+}=z_{k_0l_0}^{+}$. In particular, for~$a_+/a_-\notin\Q$ all poles and zeros are simple. 

Finally, we specify the asymptotic behavior of $G(z)$ for $\re (z)\to\pm \infty$. To this end we introduce a function $f(a_+,a_-;z)$ by setting
\be\label{Gas}
	G(a_+,a_-;z) = \exp \big(\mp i\left(\chi+\alpha z^2/4+f(a_{+},a_{-};z)\right)\big),\ \ \ \pm \re z >a_l,
\ee
where
\be\label{chi}
 \chi \equiv \frac{\pi}{24}\left(\frac{a_+}{a_-} + \frac{a_-}{a_+}\right).
\ee
It then follows from Theorem~A.1 in~\cite{Rui99} that we have
\be\label{fas}
|f(a_+,a_-;z)|<C(a_+,a_-,\im z)\exp(- \alpha a_s |\re z|/2), \ \ \
|\re z|>a_l,
\ee
where $C$ is continuous on $(0,\infty)^2\times\R$. Here and in the main text, we find it
convenient to formulate bounds that are uniform on compact subsets of a given set $S$ in terms
of positive continuous functions on $S$, generically denoted by $C$.  

It follows from the above that we have
\be
|G(z)|<C_G(z)\exp(\alpha \im z |\re z|/2),\ \ \forall z\in \cD_G,
\ee
where $C_G$ is continuous on the domain
\be\label{cDG}
\cD_G\equiv \{ z\in \C \mid z\notin -i[a,\infty)\},
\ee
and satisfies
\be
C_G(z)=1+O(\exp(- \alpha a_s |\re z|/2)),\ \ \ |\re z|\to\infty.
\ee
Here, the implied constant can be chosen uniformly for $\im z$ varying over a bounded interval.

\renewcommand{\thesection}{B}
\setcounter{equation}{0}
\setcounter{theorem}{0}

\addcontentsline{toc}{section}{Appendix B. Bounds on  $G(z-ib/2)/G(z+ib/2)$ and~$w(b;z)$}
																																												
\section*{Appendix B. Bounds on $G(z-ib/2)/G(z+ib/2)$ and~$w(b;z)$}

In the main text we have occasion to use a complex parameter~$b$ whose real part is restricted to $(0,2a)$. This entails that the downward pole sequence of the function $G(z-ib/2)$ is located in the lower half $z$-plane and the upward pole sequence of~$1/G(z+ib/2)$ in the upper half $z$-plane. 
With $b$ varying over the strip
\be
S_a\equiv \{ z\in\C \mid \re z\in (0,2a)\},
\ee
the ratio function
\be\label{rdef}
r(b;z)\equiv G(z-ib/2)/G(z+ib/2),
\ee
 plays a crucial role. It is expedient to estimate it in two different ways, depending on reality properties. Below and in Sections~4--6, we often use a new parameter
 \be\label{gam}
 \gamma\equiv \alpha \re b/2=\frac{\pi\re b}{a_{+}a_-}.
\ee
\begin{lemma}
Letting $b\in (0,2a)$, we have
\be\label{rbr}
|r(b;x)|\le  C_r(b)/\cosh(\alpha b x/2),\ \ \ \forall x\in\R,
\ee
with $C_r(b)$ a continuous function on $(0,2a)$. Moreover, defining
\be\label{cDr}
\cD_r\equiv \{ (b,z)\in S_a\times\C\mid \pm z-ib/2\notin -i[a,\infty)\},
\ee
we have
\be\label{rbc}
|r(b;z)|<C_r(b,z)/\cosh (\gamma\, \re z),\ \ \ \forall (b,z)\in\cD_r,
\ee
where $C_r(b,z)$ is a continuous function on $\cD_r$ satisfying
\be\label{Cras}
C_r(b,z)=1/2+O(\exp(- \alpha a_s |\re z|/2)),\ \ \ |\re z|\to\infty,
\ee
with the implied constant uniform for $(b,\im z)$ varying over compact subsets of~$S_a\times\R$.
\end{lemma}
\begin{proof}
From~\eqref{Gas} we obtain
\be
r(b;z)=\exp\big(\mp[ \alpha b z/2 +if(z-ib/2)-if(z+ib/2)]\big),\ \ \ \pm \re z>|\im b|+a_l.
\ee
In view of \eqref{fas} this entails
\be\label{ras}
 2\cosh(\gamma\, \re z)r(b;z)=1+O(\exp(- \alpha a_s |\re z|/2)),\ \ \ |\re z|\to\infty,
\ee
 uniformly for $(b,\im z)$ in compacts of~$S_a\times\R$. Hence, fixing $b\in (0,2a)$, we readily obtain~\eqref{rbr}.
  Since $r(b;z)$ is analytic in $\cD_r$, the second assertion easily follows from~\eqref{ras}, too.
\end{proof}

Next, we consider the weight function
\be
w(b;z)\equiv \prod_{\de=+,-}G(\de z+ia)/G(\de z+ia-ib).
\ee
By \eqref{Gcon} it is real for $b$ and $z$ real, and positive for $z\in\R^*$. Furthermore, it is even and holomorphic in
\be\label{cDw}
\cD_w\equiv \{ (b,z)\in S_a\times\C\mid \pm z+ib\notin -i[0,\infty), \pm z\notin -i[2a,\infty)\},
\ee
and on account of~\eqref{refl} and~\eqref{Gres} it satisfies
\be\label{wzero}
w(b;z)=\frac{4\pi^2}{a_+a_-}G(ib-ia)^2z^2+O(z^4),\ \ \ b\in S_a,\ \ \ z\to 0.
\ee
Just as for $r(b;z)$, we obtain two different bounds on the weight function, cf.~Lemma~B.1.

\begin{lemma}
Letting $b\in S_a$, we have
\be\label{wbr}
|w(b;x)|<C_w(b)\sinh^2(\gamma\,x),\ \ \ \forall x\in\R,
\ee
with $C_w(b)$ continuous on $S_a$.
Moreover, we have
\be\label{wbc}
|w(b;z)|<C_w(b,z)\cosh^2 (\gamma\, \re z),\ \ \ \forall (b,z)\in\cD_w,
\ee
where $C_w(b,z)$ is a continuous function on $\cD_w$~\eqref{cDw}, which satisfies
\be\label{Cwas}
C_w(b,z)=4+O(\exp(- \alpha a_s |\re z|/2)),\ \ \ |\re z|\to\infty,
\ee
 uniformly for $(b,\im z)$ varying over compacts of~$S_a\times\R$.
\end{lemma}
\begin{proof}
Letting $\pm \re z>|\im b|+a+a_l$, we get from~\eqref{Gas} and~\eqref{refl}
\be
w(b;z)=\exp\Big(\pm\Big[ \alpha bz -i\sum_{\de=+,-}\de \big[f(z+\de ia)+f(z+\de(ib-ia))\big]\Big]\Big).
\ee
Hence the assertions readily follow by using~\eqref{fas} and~\eqref{wzero}.
\end{proof}

In Section~5 (below Prop.~5.4) we employ contours $\cC$ in the $z$-plane of the form
\be
\cC =\{ s+ic +i f(s) \mid s\in\R, c\in\R,  f\in \cL\},
\ee
where $\cL$ is the space of real-valued, continuous, piecewise linear functions with compact support. There are further restrictions on the functions ensuring that the contours stay away from singularities. In particular, focussing on $w(b;z_1-z_2)$ for~$b\in S_a$ and $z_1,z_2\in \cC$, we need $\pm (z_1-z_2)+ib$ to stay away from the pole sequences in~$-i[0,\infty)$.  Therefore, we restrict attention to the subset~$\cL_w$ of  functions in $\cL$ that satisfy
\be\label{freq}
\pm (s_1-s_2)=\im b \Rightarrow \pm (f(s_1)-f(s_2))\notin (-\infty, -\re b\,],\ \ \ b\in S_a.
\ee
Note that the pole sequences arising for $\pm (z_1-z_2)\in -i[2a,\infty)$ are always avoided. Indeed, $s_1=s_2$ implies $z_1=z_2$, since $z_1$ and $z_2$ belong to the same contour.


\renewcommand{\thesection}{C}
\setcounter{equation}{0}
\setcounter{theorem}{0}

\addcontentsline{toc}{section}{Appendix C. Three explicit integrals}
																																												
\section*{Appendix C. Three explicit integrals}

The integral in the first lemma enables us to render the scheme explicit in the free cases~$b=a_{\pm}$, cf.~Section~3. The next two lemmas may be viewed as corollaries. Lemma~C.2 is a key tool to bound the joint eigenfunctions $J_N$ recursively for real $y$, yielding holomorphy in the variable $x$ in suitable domains, cf.~Props.~6.1 and~6.3. Lemma~C.3 enables us to include holomorphy properties in the variable $y$, cf.~Props.~5.4 and~6.4.

\begin{lemma}
The integral
\be\label{integral}
I_N(p,t)\equiv\int_{\mathbb{R}^N}dz\exp(iz\cdot p)\frac{\prod_{1\leq j<k\leq N}\sinh(z_j-z_k)}{\prod_{j=1}^{N+1}\prod_{k=1}^N\cosh(t_j-z_k)},
\ee
converges absolutely for all $p\in\C^N$ and $t\in\C^{N+1}$ satisfying
\be\label{ptres}
|\im p_i|<2,\ \ \ i=1,\ldots,N,\ \ \ |\im t_j|<\pi/2,\ \ \ j=1,\ldots,N+1.
\ee
For such $p,t$, with in addition $p_i\ne 0$, $i=1,\ldots,N$, and $t_j\ne t_k$ for $1\le j< k\le N+1$, it
is given by
\be\label{IN}
I_N(p,t) = \prod_{j=1}^N\frac{-i\pi}{\sinh (\pi p_j/2)}\prod_{1\leq j<k\leq N+1}\frac{1}{\sinh(t_j-t_k)}\sum_{\tau\in S_{N+1}}(-)^\tau\exp\Big(i\sum_{j=1}^Nt_{\tau(j)}p_j\Big).
\ee
\end{lemma}

\begin{proof}
The first assertion readily follows from the bound 
\be
\Big| \prod_{1\leq j<k\leq N}\sinh(z_j-z_k)\Big|< \exp((N-1)(|z_1|+\cdots +|z_N|)),\ \ \ \forall z\in \R^N.
\ee
In order to prove the explicit evaluation formula~\eqref{IN}, we first observe that
by Fubini's theorem the value of $I_N$ is independent of the order of the multiple integration. Next, starting with the $z_N$-integral, we note the equality
\begin{multline}
\sinh (\pi p_N/2) \int_{\mathbb{R}}dz_N\exp(iz_Np_N)\frac{\prod_{j=1}^{N-1}\sinh(z_j-z_N)}{\prod_{j=1}^{N+1}\cosh(t_j-z_N)}\\ =\frac{1}{2}\sum_{s=+,-}s\int_{\mathbb{R}}dz_N\exp\big(i(z_N-is\pi /2)p_N\big)\frac{\prod_{j=1}^{N-1}\sinh(z_j-z_N)}{\prod_{j=1}^{N+1}\cosh(t_j-z_N)}.
\end{multline}
Changing variable $z_N\to z_N+is\pi /2$, the right-hand side becomes
\be
-\frac{1}{2}\sum_{s=+,-}s\int_{\mathbb{R}-is\pi /2}dz_N\exp(iz_Np_N)\frac{\prod_{j=1}^{N-1}\cosh(z_j-z_N)}{\prod_{j=1}^{N+1}\sinh(t_j-z_N)}.
\ee
In the strip $|\im z_N|\leq \pi/2$, the only poles of the integrand are the simple poles located at $z_N=t_j$, $j=1,\ldots,N+1$. Invoking Cauchy's theorem, we arrive at the sum
\be
i\pi \sum_{l_N=1}^{N+1}\exp(it_{l_N}p_N)\frac{\prod_{j=1}^{N-1}\cosh(z_j-t_{l_N})}{\prod_{j\neq l_N}\sinh(t_j-t_{l_N})}.
\ee

With the $z_N$-integral now determined, we obtain
\be
I_N=\frac{i\pi}{\sinh (\pi p_N/2)}\sum_{l_N=1}^{N+1}\frac{\exp(it_{l_N}p_N)}{\prod_{j_N\neq l_N}\sinh(t_{j_N}-t_{l_N})}I_{N-1}^{(l_N)},
\ee
where
\be
I_{N-1}^{(l_N)}\equiv \int_{\mathbb{R}^{N-1}}dz_1\cdots dz_{N-1}\exp\left(i\sum_{j=1}^{N-1}z_jp_j\right)\frac{\prod_{1\leq j<k\leq N-1}\sinh(z_j-z_k)}{\prod_{j\neq l_N}\prod_{k=1}^{N-1}\cosh(t_j-z_k)}.
\ee
Since $I_{N-1}^{(l_N)}$ is of the same form as $I_{N-1}$, we can proceed recursively to compute the remaining integrals. In this way we end up with
\be
I_N=\prod_{j=1}^N\frac{i\pi}{\sinh (\pi p_j/2)}\sum{^\prime}\frac{\exp\left(i\sum_{j=1}^Nt_{l_j}p_j\right)}{P(l_1,\ldots,l_N;t)},
\ee
where
\be
P(l_1,\ldots,l_N;t)\equiv 
\prod_{j_N\neq l_N}\sinh(t_{j_N}-t_{l_N})\cdots\prod_{j_1\neq l_N,\ldots,l_1}\sinh(t_{j_1}-t_{l_1}),
\ee
and the prime signifies that the sum extends over
\be
l_N=1,\ldots,N+1, l_{N-1}\neq l_N,\ldots, l_1\neq l_N,\ldots,l_2.
\ee
To each such set of indices $l_1,\ldots,l_N$, we now associate a uniquely determined permutation $\tau\in S_{N+1}$ by requiring $l_j=\tau(j)$, $j=1,\ldots,N$.
Clearly, we have
\be\label{Pprod}
P(\tau(1),\ldots,\tau(N);t)=s(\tau) \prod_{1\leq j<k\leq N+1}\sinh(t_j-t_k),
\ee
where $s(\tau)$ equals 1 or $-1$. The product on the right-hand side of~\eqref{Pprod} yields a function that is antisymmetric in $t$, whereas $I_N(p,t)$ is manifestly symmetric in $t$. Thus $s(\tau)$ must be equal to $(-)^{\tau}$ up to an overall sign. Letting $\tau= {\rm id}$, we see this sign equals~$(-)^N$, so that we arrive at~\eqref{IN}.
\end{proof}

The absolute convergence of~\eqref{integral} for all~$(p,t)\in\C^N\times\C^{N+1}$ in the polystrips~\eqref{ptres} implies holomorphy of $I(p,t)$ in the resulting product domain, but there is no useful closed formula for the coincidence limits of the right-hand side of~\eqref{IN}. Also, we can freely take $p_i$-partials under the integral sign in view of the exponential decay of the integrand for $z_i\to\pm\infty$. In particular, in Section~5 we shall encounter an integral
\be\label{B2}
B_2(t)\equiv\int_{\R^2}dz\frac{(z_1-z_2)\sinh  (z_1-z_2)}{\prod_{j=1}^{3}\prod_{k=1}^2\cosh(t_j-z_k)}=\lim_{p\to 0}(-i\partial_{p_1}+i\partial_{p_2})I_2(p,t),
\ee
whose direct calculation from~\eqref{IN} would already be unwieldy.  

In order to avoid this chore, we introduce
\be
F(p,t)\equiv \sum_{\tau\in S_3}(-)^{\tau}\exp (i\tau(t)\cdot p),
\ee
where both $p$ and $t$ belong to $\C^3$. By antisymmetry of $F$ in $p$ and $t$, its power series expansion is of the form
\be
F(p,t)=G(p,t)\prod_{1\le m<n\le 3}(p_m-p_n)(t_m-t_n),
\ee
where $G(p,t)$ is symmetric in $p$ and $t$, and invariant under interchange of $p$ and $t$. Furthermore, we have
\be
G(p,t)=c_0 +\sum_{j,k=1}^3 c_{jk}p_jt_k +O((p_mt_n)^2),
\ee
where the notation will be clear from context. It is not hard to verify that~$c_0$ equals~$-i/2$, so that we can now deduce
\be\label{Frep}
F((p_1,p_2,0),t)=-\frac{i}{2}p_1p_2(p_1-p_2)\prod_{1\le m<n\le 3}(t_m-t_n)(1+R(p_1,p_2,t)),
\ee
where the remainder $R$ vanishes for $p_1=p_2=0$.
When we now use the representation~\eqref{Frep} to rewrite $I_2(p,t)$ as specified by~\eqref{IN}, then a moment's thought shows that the integral~\eqref{B2} is given by
\be
B_2(t)=4 \prod_{1\le m<n\le 3}\frac{t_m-t_n}{\sinh  (t_m-t_n)}.
\ee

The following lemma details the arbitrary-$N$ version of this integral, which we need in Section~6.

\begin{lemma}
Let  $t\in\C^{N+1}$, with $|\im t_j|<\pi /2$, $j=1,\ldots,N+1$. Then the integral
\be\label{Bintegral}
B_N(t)\equiv
\int_{\R^N}dz\frac{\prod_{1\le m<n\le N}(z_m-z_n)\sinh  (z_m-z_n)}{\prod_{j=1}^{N+1}\prod_{k=1}^N\cosh(t_j-z_k)},
\ee
is given by
\be\label{BN}
B_N(t) =2^N  \prod_{1\le m<n\le N+1}\frac{t_m-t_n}{\sinh  (t_m-t_n)}.
\ee
\end{lemma}

\begin{proof}
We shall proceed in a slightly different manner compared to the $N=2$ case considered above. However, the starting point is the same, namely, the function
\be\label{FN}
F_N(p,z)\equiv \sum_{\tau\in S_N}(-)^{\tau} \exp(i\tau(z)\cdot p),
\ee
where $p,z\in\C^N$. Reasoning as before, we find that its power series expansion is of the form
\be\label{FNpow}
\left(\eta_N+\sum_{j,k=1}^Nc_{jk}p_jz_k + O\big((p_mz_n)^2\big)\right)\prod_{1\leq m<n\leq N}(p_m-p_n)(z_m-z_n).
\ee
From this we deduce
\begin{multline}\label{FN0}
F_N((p_1,\ldots,p_{N-1},0),z)\\ = \eta_N p_1\cdots p_{N-1}\prod_{1\leq j<k\leq N-1}(p_j-p_k)\prod_{1\leq m<n\leq N}(z_m-z_n)\cdot\big(1+R(p,z)\big),
\end{multline}
where the remainder $R(p,z)$ vanishes at $p=0$. To determine the constant $\eta_N$, we note that the monomial $(p_1z_1)^{N-1}(p_2z_2)^{N-2}\cdots p_{N-1}z_{N-1}$ only occurs in the expansion of
\begin{multline}
\frac{1}{(N(N-1)/2)!}(ip_1z_1+\cdots +ip_Nz_N)^{N(N-1)/2}\\ = i^{N(N-1)/2}\sum_{\substack{j_1,\ldots,j_N=0\\ j_1+\cdots+j_N=N(N-1)/2}}^{N(N-1)/2}\frac{(p_1z_1)^{j_1}\cdots(p_Nz_N)^{j_N}}{j_1!\cdots j_N!}.
\end{multline}
Consequently, we have
\be
\eta_N=\frac{i^{N(N-1)/2}}{(N-1)!(N-2)!\cdots2!}.
\ee

Now from~\eqref{FNpow} we see that the integral $B_N(t)$ is given by
\be\label{BNlim}
B_N(t) = \lim_{p\to 0}\eta_N^{-1}\frac{\sum_{\sigma\in S_N}(-)^{\sigma}I_N(\sigma(p),t)}{\prod_{1\leq m<n\leq N}(p_m-p_n)}.
\ee
On the other hand, requiring at first $t_j\neq t_k$ for $1\leq j<k\leq N+1$, and using $\tau(t)\cdot\sigma(p)=(\sigma^{-1}\circ\tau)(t)\cdot p$, $\sigma\in S_N$, we deduce from \eqref{IN} that the numerator can be written
\begin{multline}
\sum_{\sigma\in S_N}(-)^{\sigma}I_N(\sigma(p),t)\\ = N!F_{N+1}\big((p_1,\ldots,p_N,0),t\big)\prod_{j=1}^N\frac{-i\pi}{\sinh (\pi p_j/2)}\prod_{1\leq m<n\leq N+1}\frac{1}{\sinh(t_m-t_n)}.
\end{multline}
It follows from \eqref{FN0} that the rhs of this equality is of the form
\be
(-2i)^NN!\eta_{N+1}\prod_{1\leq j<k\leq N}(p_j-p_k)\prod_{1\leq m<n\leq N+1}\frac{t_m-t_n}{\sinh(t_m-t_n)}\cdot\big(1+\tilde{R}(p,t)\big),
\ee
where again the remainder $\tilde{R}(p,t)$ vanishes at $p=0$. Substituting this expression for the numerator and observing that $(-i)^NN!\eta_{N+1}=\eta_N$, we arrive at \eqref{BN}. Since both the integral and the rhs of~\eqref{BN} are well defined when some coordinates coincide, the evaluation formula follows for all $t$ in the specified polystrip.
\end{proof}

To state the final lemma (which we invoke for Props.~5.4 and 6.4), we introduce
\be
Z_N\equiv \frac1N\sum_{j=1}^Nz_j,\ \ \ z\in\R^N.
\ee

\begin{lemma}
The integral
\be\label{Cintegral}
C_N(u,t)  \equiv  \int_{\mathbb{R}^N}dz\frac{\prod_{1\leq j<k\leq N}\sinh(z_j-z_k)}{\prod_{j=1}^{N+1}\prod_{k=1}^N\cosh(t_j-z_k)}\exp\Big( \sum_{j=1}^N z_j(u_{N+1}-u_j)\Big),
\ee
converges absolutely for all $(u,t)$ in the domain
\be\label{cCN}
\cC_N\equiv \{ u,t\in\C^{N+1} \mid  |\re (u_{N+1}-u_i)|<2,\  i=1,\ldots, N,\ |\im t_j|<\pi /2,\  j=1,\ldots,N+1\}.
\ee
The quotient function 
\be\label{QN}
Q_N(u,t)\equiv C_N(u,t)/\prod_{1\le j<k\le N}(u_j-u_k),
\ee
is holomorphic in~$\cC_N$.
For $(u,t)\in\cC_N$, with in addition $u_i\ne u_{N+1}$ for $i=1,\ldots,N$, and $t_j\ne t_k$ for $1\le j< k\le N+1$, we have
\bea\label{CN}
C_N(u,t) & =  & \prod_{j=1}^N\frac{\pi}{\sin (\pi (u_{N+1}-u_j)/2)}\prod_{1\leq j<k\leq N+1}\frac{1}{\sinh(t_j-t_k)}\nonumber \\
&  &  \times \sum_{\tau\in S_{N+1}}(-)^\tau\exp\Big(\sum_{j=1}^Nt_{\tau(j)}(u_{N+1}-u_j)\Big).
\eea
 \end{lemma}
\begin{proof}
Comparing~\eqref{Cintegral} to~\eqref{integral},  we deduce
\be
C_N(u,t)=I_N(p(u),t),\ \ \ p(u)_j=i(u_j-u_{N+1}),\ \ \ j=1,\ldots,N.
\ee
Hence the first assertion and the explicit evaluation~\eqref{CN} follow from Lemma~C.1. The evaluation reveals that $C_N(u,t)$ is antisymmetric under permutations of $u_1,\ldots, u_N$, so that holomorphy of $Q_N$ in~$\cC_N$ readily follows.
\end{proof}

At the end of Section~6 we are in the position to deduce a remarkable positivity feature of~$Q_N(u,t)$ from the use of Lemma~C.3 for obtaining the bound~\eqref{JNby}, cf.~\eqref{QNpos}. Under the stronger assumption~$u_j\in(u_{N+1}-2,u_{N+1})$, $j=1,\ldots,N$, this feature can also be derived directly from~Lemma~C.3 (using finite induction). However, due to the special role of $u_{N+1}$, a direct proof of~\eqref{QNpos} in the setting of this appendix would be quite laborious.


\bibliographystyle{amsalpha}

\end{document}